\RequirePackage{etoolbox}

\newbool{preprint}
\newbool{ipco}  

\setbool{preprint}{true}

\documentclass{article}

\usepackage[utf8]{inputenc}

\usepackage[letterpaper,margin=1in]{geometry}
\usepackage{amsmath, amsthm}

\usepackage{amssymb}

\usepackage{complexity}

\usepackage{dsfont} %
\usepackage{bbm} %

\usepackage{hyperref}
\hypersetup{colorlinks, linkcolor=darkblue, citecolor=darkgreen, urlcolor=darkblue}
\usepackage[capitalize, nameinlink]{cleveref}
\crefname{theorem}{Theorem}{Theorems}
\Crefname{lemma}{Lemma}{Lemmas}
\Crefname{claim}{Claim}{Claims}
\Crefname{fact}{Fact}{Facts}
\Crefname{remark}{Remark}{Remarks}
\Crefname{observation}{Observation}{Observations}
\Crefname{figure}{Figure}{Figures}
\Crefname{line}{Line}{Lines}
\Crefname{algocf}{Algorithm}{Algorithms}
\crefalias{AlgoLine}{line}
\Crefname{stepsromani}{Step}{Steps}
\Crefname{stepsarabici}{Step}{Steps}

\newtheorem{theorem}{Theorem}
\newtheorem{lemma}[theorem]{Lemma}

\newtheorem{definition}[theorem]{Definition}
\newtheorem{remark}[theorem]{Remark}
\newtheorem{corollary}[theorem]{Corollary}

\newtheorem{observation}[theorem]{Observation}
\newtheorem{claim}[theorem]{Claim}

\usepackage{times}

\usepackage[T1]{fontenc}
\usepackage{fix-cm}

\usepackage[inline]{enumitem}
\setlist[enumerate,1]{label=(\roman*), leftmargin=2.2em}
\setlist[enumerate,2]{label=(\alph*)}
\setlist{nosep,topsep=0.1em}
\setlist[itemize,1]{label={\bfseries--}}

\newlist{stepsroman}{enumerate}{1}
\setlist[stepsroman]{label=(\roman*), leftmargin=2.2em}

\newlist{stepsarabic}{enumerate}{1}
\setlist[stepsarabic]{rightmargin=0.2em, label=\arabic*., ref=\arabic*}

\usepackage{mathtools}

\usepackage[linewidth=0.8pt]{mdframed}
\mdfsetup{skipabove=\bigskipamount, skipbelow=\bigskipamount, innerleftmargin=0.5em, innertopmargin=0.5em, innerrightmargin=0.5em, innerbottommargin=0.5em, align=center}

\usepackage[color=darkgreen!40!white]{todonotes}
\setuptodonotes{size=\scriptsize}

\usepackage{xcolor}
\definecolor{darkblue}{rgb}{0,0,0.38}
\definecolor{darkred}{rgb}{0.6,0,0}
\definecolor{darkgreen}{rgb}{0.1,0.35,0}

\usepackage[
backend=biber,
style=alphabetic,
citestyle=alphabetic,
maxalphanames=6,
maxcitenames=99,
mincitenames=98,
maxbibnames=99,
giveninits=true,
url=false,
doi=true,
isbn=true,
backref=true
]{biblatex}

\usepackage{xpatch}

\makeatletter
\patchcmd\blx@bblinput{\blx@blxinit}
                      {\blx@blxinit
                      }{}{\fail}
\makeatother

\newcommand{\footremember}[2]{%
    \footnote{#2}
    \newcounter{#1}
    \setcounter{#1}{\value{footnote}}%
}
\newcommand{\footrecall}[1]{%
    \footnotemark[\value{#1}]%
}

\usepackage{xspace}

\makeatletter
\newcommand{\linkdest}[1]{\Hy@raisedlink{\hypertarget{#1}{}}}
\makeatother

\newcommand{\conv}{\operatorname{conv}}
\newcommand{\linspan}{\operatorname{span}}

\newcommand{\Lall}{\mathcal{L}_{\mathrm{all}}}

\newcommand{\maxcut}{\protect\hyperlink{prb:max-cut}{\textsc{Max-Cut}}\xspace}
\newcommand{\cubicmaxcut}{\protect\hyperlink{prb:cubic-max-cut}{\textsc{Cubic-Max-Cut}}\xspace}
\newcommand{\domP}{P^{\uparrow}_{\textup{odd}}}
\newcommand{\domCut}{P^{\uparrow}_{\textup{cut}}}

\makeatletter
\DeclareFieldFormat{eprint:arxiv}{%
  arXiv\addcolon\space
  \ifhyperref
    {\href{https://arxiv.org/abs/#1}{%
       \nolinkurl{#1}%
       \iffieldundef{eprintclass}
     {}
     {\addspace\mkbibbrackets{\thefield{eprintclass}}}}}
    {\nolinkurl{#1}
     \iffieldundef{eprintclass}
       {}
       {\addspace\mkbibbrackets{\thefield{eprintclass}}}}}
\makeatother

\usepackage{wrapfig}

\usepackage[margin=25pt,font=small,labelfont=bf]{caption}

\usepackage{subcaption}

\usepackage{graphicx}

\graphicspath{{.}{graphics/}}
\makeatletter
\newcommand\appendtographicspath[1]{%
  \g@addto@macro\Ginput@path{#1}%
}
\makeatother

\usepackage{tikz}
\usetikzlibrary{shapes.geometric}
\usetikzlibrary{calc}

\usepackage[vlined,ruled,algo2e]{algorithm2e}
\SetAlCapHSkip{0.2em}
\SetAlgoInsideSkip{smallskip}
\AtBeginDocument{%
  \SetCustomAlgoRuledWidth{0.95\textwidth}%
  \setlength{\algowidth}{0.975\textwidth}%
}
\makeatletter
\patchcmd{\@algocf@start}
  {\begin{lrbox}{\algocf@algobox}}
  {%
   \hspace*{0.025\textwidth}%
   \begin{lrbox}{\algocf@algobox}%
   \begin{minipage}{0.95\textwidth}%
  }
  {}{}
\patchcmd{\@algocf@finish}
  {\end{lrbox}}
  {\end{minipage}\end{lrbox}}
  {}{}
\makeatother

\usepackage{xfrac}

\let\oldtop\top
\renewcommand{\top}{{\scriptscriptstyle{\oldtop}}}

\makeatletter
\def\@fnsymbol#1{\ensuremath{\ifcase#1\or *\or %
\ddagger\or
    \mathsection\or \mathparagraph\or \|\or **\or \dagger\dagger
    \or \ddagger\ddagger \else\@ctrerr\fi}}
\makeatother

\makeatletter
\let\save@mathaccent\mathaccent
\newcommand*\if@single[3]{%
  \setbox0\hbox{${\mathaccent"0362{#1}}^H$}%
  \setbox2\hbox{${\mathaccent"0362{\kern0pt#1}}^H$}%
  \ifdim\ht0=\ht2 #3\else #2\fi
}
\newcommand*\rel@kern[1]{\kern#1\dimexpr\macc@kerna}
\newcommand*\widebar[1]{\@ifnextchar^{{\wide@bar{#1}{0}}}{\wide@bar{#1}{1}}}
\newcommand*\wide@bar[2]{%
  \if@single{#1}{\wide@bar@{#1}{#2}{1}}{\wide@bar@{#1}{#2}{2}}%
}
\newcommand*\wide@bar@[3]{%
  \begingroup
  \def\mathaccent##1##2{%
    \let\mathaccent\save@mathaccent
    \if#32 \let\macc@nucleus\first@char \fi
    \setbox\z@\hbox{$\macc@style{\macc@nucleus}_{}$}%
    \setbox\tw@\hbox{$\macc@style{\macc@nucleus}{}_{}$}%
    \dimen@\wd\tw@
    \advance\dimen@-\wd\z@
    \divide\dimen@ 3
    \@tempdima\wd\tw@
    \advance\@tempdima-\scriptspace
    \divide\@tempdima 10
    \advance\dimen@-\@tempdima
    \ifdim\dimen@>\z@ \dimen@0pt\fi
    \rel@kern{0.6}\kern-\dimen@
    \if#31
    \overline{\rel@kern{-0.6}\kern\dimen@\macc@nucleus\rel@kern{0.4}\kern\dimen@}%
    \advance\dimen@0.4\dimexpr\macc@kerna
    \let\final@kern#2%
    \ifdim\dimen@<\z@ \let\final@kern1\fi
    \if\final@kern1 \kern-\dimen@\fi
    \else
    \overline{\rel@kern{-0.6}\kern\dimen@#1}%
    \fi
  }%
  \macc@depth\@ne
  \let\math@bgroup\@empty \let\math@egroup\macc@set@skewchar
  \mathsurround\z@ \frozen@everymath{\mathgroup\macc@group\relax}%
  \macc@set@skewchar\relax
  \let\mathaccentV\macc@nested@a
  \if#31
  \macc@nested@a\relax111{#1}%
  \else
  \def\gobble@till@marker##1\endmarker{}%
  \futurelet\first@char\gobble@till@marker#1\endmarker
  \ifcat\noexpand\first@char A\else
  \def\first@char{}%
  \fi
  \macc@nested@a\relax111{\first@char}%
  \fi
  \endgroup
}
\makeatother
\makeatletter
\NewCommandCopy\@@pmod\pmod
\DeclareRobustCommand{\pmod}{\@ifstar\@pmods\@@pmod}
\def\@pmods#1{\mkern4mu({\operator@font mod}\mkern 6mu#1)}
\makeatother
\tikzset{highlight/.style={line width=4pt, gray!60}}
\tikzset{node_style/.style={circle, thick, draw, fill=white, minimum size=6pt, inner sep=0pt}}
\tikzset{standard_line/.style={line width=1pt}}
 
\title{On the Complexity of the Odd-Red Bipartite Perfect Matching Polytope%
  \texorpdfstring{%
    \thanks{%
        This work was partially supported by the Swiss National Science
        Foundation (grant no.\ P500PT\_206742) and the Deutsche
        Forschungsgemeinschaft (DFG, German Research Foundation) under
        Germany's Excellence Strategy~--~EXZ-2047/1~--~390685813.
    }%
  }{}
}

\author{%
  Martin Nägele%
  \footremember{ETH}{%
    Department of Mathematics, ETH Zurich, Switzerland.
    Email:
    ${\{\text{\href{mailto:martinn@ethz.ch}{martinn},
          \href{mailto:cnoebel@ethz.ch}{cnoebel},
    \href{mailto:ricoz@ethz.ch}{ricoz}}\}}$@ethz.ch.
    Parts of this work were done while M.~Nägele was employed at
    University of Bonn.%
  }
  \and
  Christian Nöbel%
  \footrecall{ETH}
  \and
  Rico Zenklusen%
  \footrecall{ETH}
}%
\date{}

\begin{document}
  \renewcommand{\coloneqq}{:=}
  \renewcommand{\eqqcolon}{=:}

\maketitle

\begin{abstract}
The odd-red bipartite perfect matching problem asks to find a perfect matching containing an odd number of red edges in a given red-blue edge-colored bipartite graph.
While this problem lies in $\P$, its polyhedral structure remains elusive---despite renewed attention to achieving better polyhedral understanding, nurtured by recent advances from two complementary angles.
Apart from being a special case of bimodular integer programs, whose polyhedral structure is also badly understood, it is related to one of the most notorious open derandomization questions in theoretical computer science: whether there is a deterministic efficient algorithm for the exact bipartite perfect matching problem, which asks to find a perfect matching with exactly $k$ red edges.

Recent progress towards deterministic algorithms for this problem crucially relies on a good polyhedral understanding.
Motivated by this, Jia, Svensson, and Yuan show that the extension complexity of the exact bipartite perfect matching polytope is exponential in general.
Interestingly, their result is true even for the easier odd-red bipartite perfect matching problem.
For this problem, they introduce an exponential-size relaxation and leave open whether it is an exact description.

Apart from showing that this description is not exact and even hard to separate over, we show, more importantly, that the red-odd bipartite perfect matching polytope exhibits complex facet structure:
any exact description needs constraints with large and diverse coefficients.
This rules out classical relaxations based on constraints with all coefficients in~$\{0,\pm1\}$, such as the above-mentioned one, and suggests that significant deviations from prior approaches may be needed to obtain an exact description.
More generally, we obtain that also polytopes corresponding to bimodular integer programs have complex facet structure.
 \end{abstract}

  \begin{tikzpicture}[overlay, remember picture, shift = {(current
    page.south east)}]
    \node[anchor=south east, outer sep=5mm] {
      \begin{tikzpicture}[outer sep=0] %
        \node (dfg) {\includegraphics[height=12mm]{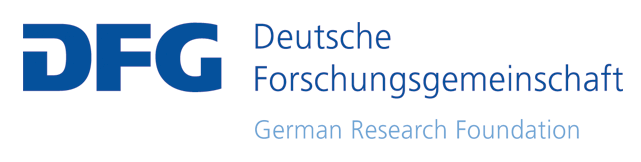}};
        \node[left=5mm of dfg, yshift=1mm] (snf)
        {\includegraphics[height=8mm]{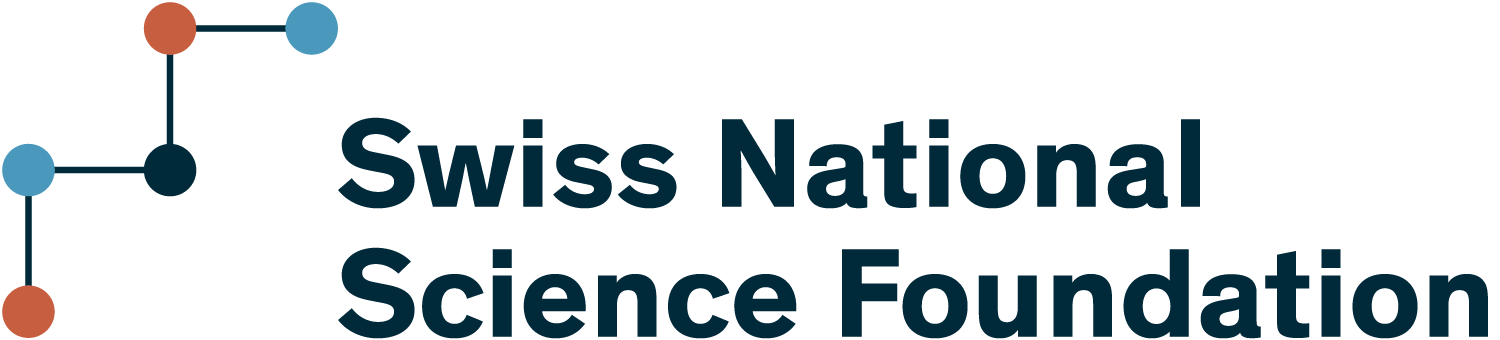}};
      \end{tikzpicture}
    };
  \end{tikzpicture}

  \thispagestyle{empty}
\addtocounter{page}{-1}
\newpage

\section{Introduction}

Parity-constrained variations of classical combinatorial optimization problems have been studied in the literature for a long time.
The arguably most prominent efficiently solvable example is the minimum odd cut problem, which was first discussed by \textcite{padberg_1982_odd}: Here, the goal is to find a cut $C\subseteq V$ minimizing $|\delta(C)|$, i.e., the number of edges crossing the cut, under the extra constraint that $|C|$ is odd.%
\footnote{
    Given a graph $G=(V,E)$, we denote for $C\subseteq V$ by $\delta(C)$ the set of edges $e\in E$ with exactly one endpoint in $C$. For $v\in V$, we use the shorthand notation $\delta(v) \coloneqq \delta(\{v\})$.
}
We consider another efficiently solvable%
\footnote{%
    The bipartite odd-red perfect matching problem can be solved in polynomial time by starting with an arbitraty perfect matching and, if necessary, swapping the role of matching and non-matching edges on an alternating cycle with an odd number of red edges.
    This has, e.g., been observerd in \cite{artmann_2017_bimodular}.
}
parity-constrained problem, namely odd-red perfect matching in bipartite graphs, defined as follows.

\begin{mdframed}[userdefinedwidth=0.95\linewidth]
\linkdest{prb:BORPM}{\textbf{Bipartite Odd-Red Perfect Matching:}}
    Given a bipartite graph $G=(V,E)$, a subset $R\subseteq E$ of whose edges are red, the task is to decide whether there exists a perfect matching $M\subseteq E$ in $G$ with $|M\cap R|$ odd.
\end{mdframed}
This problem has gained increased attention recently due to its connections to both Integer Programming with bounded subdeterminants and derandomization questions in Theoretical Computer Science, on which we elaborate next.

Recently, parity-constrained problems have come up in an interesting connection to a well-known conjecture in Integer Programming, claiming efficient solvability of integer programs when the constraint matrix fulfills certain subdeterminant bounds.
More precisely, an IP $\min\{ c^\top x\colon Ax\leq b, x\in\mathbb{Z}^n_{\geq 0} \}$ is called \emph{$\Delta$-modular} (for some $\Delta\in \mathbb{Z}_{\geq 1}$) if the constraint matrix $A$ has full column rank and all its $n\times n$ subdeterminants are bounded in absolute value by $\Delta$.
It is conjectured that $\Delta$-modular IPs can be solved in polynomial time for every fixed~$\Delta$.%
\footnote{When talking about polynomial time procedures, we use the standard convention that these procedures are deterministic. When dealing with randomized procedures, we explicitly mention this.}
One can see that this conjecture implies another well-known variant of it, where instead of requiring $A$ to be $\Delta$-modular, one requires that all subdeterminants of $A$ (not only the $n\times n$ ones) are bounded in absolute value by~$\Delta$ (these are so-called \emph{totally $\Delta$-modular matrices}).%
\footnote{Hence, this version does not need the constraint matrix $A$ to have full column rank.
Concretely, a totally $\Delta$-modular problem can be interpreded as a $\Delta$-modular problem with full-rank constraint matrix by enlarging the inequality system $Ax\leq b$ through the addition of non-negativity constraints $x\geq 0$.
If $A$ is totally $\Delta$-modular, this modification leads to another totally $\Delta$-modular constraint matrix~$\begin{psmallmatrix}A \\ -I\end{psmallmatrix}$ with then full column rank, which is therefore also a $\Delta$-modular constraint matrix.%
}
\textcite{artmann_2017_bimodular} proved the general version of the conjecture to be true for $\Delta=2$, the so-called \emph{bimodular} case.
Their approach prominently features parity constraints: They show how to reduce bimodular IPs to IPs with totally unimodular constraint matrices and a single parity constraint, and then show that such problems can be solved efficiently.
Efforts towards generalizing the result of \citeauthor{artmann_2017_bimodular} have triggered interest in extending simple parity constraints to more general congruency constraints.
Some interesting questions in this context are, for example, asking for a minimum cut of size congruent to $k$ modulo $m$, for some given $k,m\in\mathbb{Z}$ \cite{naegele_2019_submodular,naegele_2020_newContraction}, or attempts towards solving totally unimodular IPs with general congruency constraints with modulus $m>2$, which can be seen as special cases of $m$-modular IPs \cite{naegele_2023_congruency,naegele_2023_advances}.
Additionally, congruency constraints have recently been studied in matroid base problems, addressing the question of whether, in a matroid whose elements are labelled by integers, there is a basis whose labels sum up to $r$ modulo $m$, for some given integers $r$ and $m$ \cite{liu_2024_congruency,hoersch_2024_group}.
Intriguingly, a common theme in all these extensions is that they face barriers when it comes to general moduli $m$; many approaches are limited to constant prime power moduli, or need to allow for other compromises such as obtaining approximate solutions only or having to involve randomization in the algorithms.
The unsatisfactory situation might be explained by a general lack of techniques to approach congruency-constrained problems.
Present work typically follows either a partial enumeration approach, or relies on strong proximity results between solutions of the original and the congruency-constrained problem, but these techniques seem to come with inherent limitations.
An alternative approach would be to develop a better polyhedral understanding of congruency-constrained problems, which is currently lacking.

\medskip

The limited polyhedral understanding of Odd-Red Bipartite Perfect Matchings has recently been highlighted by \textcite{jia_2023_exactMatching} in a completely different context.
They were interested in one of the most notorious open derandomization questions in Theoretical Computer Science, namely whether there is a deterministic efficient algorithm for the Bipartite Exact Perfect Matching problem, often also referred to as the Red-Blue Matching Problem.
The setting is identical to Bipartite Odd-Red Perfect Matching, except that the goal is to find a perfect matching with exactly some given number $k$ of red edges.
This problem can be solved by a celebrated randomized algorithm of \textcite{mulmuley_1987_matching} in polynomial time, but it is a major open question whether there is a deterministic polynomial-time algorithm for it.
There has been no notable progress towards derandomizing the above-mentioned approach until a relatively recent breakthrough of \textcite{fennerBipartitePerfectMatching2016} almost completely derandomized the approach of \citeauthor{mulmuley_1987_matching} for the bipartite perfect matching problem (without colors), which had interesting implications in the field of parallel algorithms.\footnote{We refer to an almost complete derandomization because the resulting deterministic parallel algorithm needs quasi-polynomially many processors.}
Building up on this, derandomizations of the technique of \citeauthor{mulmuley_1987_matching} have been developed for the perfect matching polytope for general graphs~\cite{svenssonMatchingProblemGeneral2017}, matroid intersection~\cite{gurjarLinearMatroidIntersection2020}, and totally unimodular polytopes~\cite{gurjarIsolatingVertexLattices2021}.
What is common in all of these advances is that they rely on a good polyhedral understanding of the underlying problem.
However, this is completely lacking for the Bipartite Exact Perfect Matching problem.
As a natural first step in this direction, \textcite{jia_2023_exactMatching} studied the polyhedral structure of the Bipartite Odd-Red Perfect Matching problem.
One reason why this is a natural first step is that a description of the Bipartite Exact Perfect Matching polytope could be used to obtain a description (in extended space) of the Bipartite Odd-Red Perfect Matching polytope.
This can be achieved by observing that the latter polytope is the convex hull of the union of slices of the former polytope, namely one slice per each odd number of red edges, and by using Balas' disjunctive programming approach (see~\cite{balasDisjunctiveProgramming2018}).

\citeauthor{jia_2023_exactMatching} showed that the Bipartite Exact Perfect Matching polytope has exponential extension complexity.
However, there remains hope that a good (exponential) description of it may still be helpful for derandomization purposes, the same way as it was helpful for derandomizing the approach of  \citeauthor{mulmuley_1987_matching} for perfect matchings in general graphs~\cite{svenssonMatchingProblemGeneral2017}, which is also a problem in \P\ with exponential extension complexity~\cite{rothvoss2017matching}.
To this end, \citeauthor{jia_2023_exactMatching} introduced a candidate description for the Bipartite Odd-Red Perfect Matching polytope, and showed that it has the desirable property that it is empty if and only if there is no odd-red perfect matching in the underlying graph.
Before we state our results, we briefly review this candidate description.

\subsection[The candidate description of Jia, Svensson, and Yuan]{The candidate description of \citeauthor{jia_2023_exactMatching}}
\label{sec:candidate-descripton}

Let $G=(V,E)$ be a bipartite graph with $|V|=2n$ vertices for some $n\in\mathbb{Z}_{\geq 1}$, with $n$ vertices on each side of the bipartition, and let $R\subseteq E$ be the set of red edges.
We denote by $P_{(G,R)}$ the \emph{bipartite odd-red perfect matching polytope}, i.e., 
\begin{equation*}
    P_{(G,R)} \coloneqq \conv(\{\chi^M\colon M\text{ is a perfect matching in $G$ with $|M\cap R|$ odd}\})\enspace,
\end{equation*}
where $\chi^M\in\{0,1\}^E$ is the incidence (or characteristic) vector of $M$.

\textcite{jia_2023_exactMatching} defined a family of valid constraints for $P_{(G,R)}$ based on the following.
Consider a set of $\{0, 1\}$ labelings of $V$ given by
$$
\Lall(G) \coloneqq \{L \colon V \rightarrow \{0,1\} \text{
with } |L^{-1}(1)| \equiv n \pmod*{2}\}\enspace,
$$
and for every $L\in \Lall(G)$, define
$$
  E_L = \{e\in E\setminus R \colon L(u) = L(v)\} \cup \{e \in
  R \colon L(u) \neq L(v)\}\enspace.
$$
One can observe that for each $L\in \Lall(G)$, every odd-red perfect matching must use at least one edge in $E_L$.
In other words, constraints of the form $x(E_L) \geq 1$ are valid, i.e., they will be satisfied by all incidence vectors of odd-red perfect matchings.
We refer to these constraints as \emph{label constraints}.
They lead to the following relaxation of $P_{(G,R)}$:
\begin{align}\label{eq:label-constraint-relaxation}\tag{label constraint relaxation}
    Q_{(G,R)} \coloneqq \left\{
        x \in \mathbb{R}_{\geq 0}^E \colon
        \begin{array}{rl}
          x(\delta(u)) = 1 & \ \forall u\in V\\
          x(E_L)  \geq 1 & \ \forall L \in \Lall(G)
        \end{array}
      \right\}\enspace.
\end{align}
As mentioned, it was left open by \citeauthor{jia_2023_exactMatching} whether the label constraint relaxation $Q_{(G,R)}$ is a correct description of $P_{(G,R)}$, i.e., whether $P_{(G,R)}=Q_{(G,R)}$.
Intriguingly, they showed that their relaxation has the remarkable property that $Q_{(G,R)} = \emptyset$ if and only if $P_{(G,R)} = \emptyset$.

\subsection{Our results}

In this work, we refute in a strong way the possibility that the label constraint relaxation could capture $P_{(G,R)}$ by showing that any inequality description of $P_{(G,R)}$ must be surprisingly complex.
More precisely, any description must deviate significantly not only from the proposed label constraint relaxation $Q_{(G,R)}$, but also from many other classical constraint families for combinatorial optimization problems.
Concretely, our main result exhibits an exponentially sized family of facets of $P_{(G,R)}$ that require large and diverse coefficients in any description through a linear constraint.
Note that because $P_{(G,R)}$ is a subset of the bipartite perfect matching polytope, it is not full-dimensional.
Thus, a facet of $P_{(G,R)}$ may be described by linear inequalities with different normal vectors that are not scaled versions of each other.
\cref{thm:diverse_and_large_coeff} shows that any such inequality has complicated structure.

\begin{theorem}\label{thm:diverse_and_large_coeff}
    For every even $n\in \mathbb{Z}_{>0}$, there exists a bipartite graph $G = (V,E)$ on $n$ vertices, red edges $R\subseteq E$, and an exponentially (in $n$) sized family $\mathcal{F}$ of facets of the polytope $P_{(G,R)}$ with the following property:
    For every $F\in\mathcal{F}$ and all $a\in \mathbb{Z}^E$, $b\in \mathbb{Z}$ such that $F = \{x\in P_{(G,R)}\colon a^\top x=b\}$, we have
    \[
    \max_{e\in E} |a_e| \geq \frac{n-4}{2}
    \qquad\text{and}\qquad
    \lvert\{a_e\colon e\in E\}\rvert\geq \sqrt{\frac{n-1}{2}}\enspace.
    \]
\end{theorem}

In particular, \cref{thm:diverse_and_large_coeff} excludes any inequality description of $P_{(G,R)}$ consisting of constraints with all coefficients in $\{0,\pm1\}$, as for example the candidate description by \citeauthor{jia_2023_exactMatching}.
More importantly, though, it underlines that any exact description must involve constraints distinctly different from those typically appearing in most other polytopes associated with combinatorial optimization problems.
This is despite the fact that $P_{(G,R)}$ is well-behaved from other points of view, in particular that we can optimize over it in polynomial time~\cite{artmann_2017_bimodular}.
We obtain \Cref{thm:diverse_and_large_coeff} by identifying and exploiting a connection to the dominant of the odd cycle polytope, for which we show a similar (actually even slightly stronger) facet-complexity result (see \Cref{thm:diverse_and_large_coeff_odd_cycle}).
This puts both the odd-red perfect matching problem and the minimum odd cycle problem in a rare class of natural and efficiently solvable combinatorial problems (not involving any numbers as part of their input) whose associated polytopes have highly complex inequality descriptions.
There are other notable examples of problems with that property.
This includes, in paticular, the stable set problem in claw-free graphs~\cite{pecherFacetsStableSet2007}.
Another example, this time of an unbounded polyhedron instead of a polytope, is the dominant of the cut polytope, even for restricted graph classes.%
\footnote{%
One way this is implied by prior results is as follows: The blocker of the dominant $\domCut$ of the cut polytope is the polyhedron $P = \{x\in \mathbb{R}_{\geq 0}^E\colon x(\delta(C))\geq 2 \ \forall C\subsetneq V, C\neq\emptyset\}$, which is strongly related to the subtour elimination polytope.
In particular, any vertex of $P$ is also a vertex of the subtour elimination polytope.
The subtour elimination polytope, and therefore also $P$, can have vertices with many different coefficients (this follows, for example, from \cite{boyd_thesis}).
Finally, it remains to observe that the blocker of $P$ is again $\domCut$, and each vertex of $P$ appears as the coefficient vector of a facet-defining inequality of $\domCut$.
}%

In order to identify complex facets as in \cref{thm:diverse_and_large_coeff}, it is enough to consider the bipartite graph $G_0$ with two parallel edges connecting every two vertices on different sides of the bipartition, and red edges $R_0$ containing precisely one of the two parallel edges between each such pair of vertices.
Indeed, every other bipartite graph $G$ on the same number of vertices with red edges $R$ is a subgraph of $(G_0,R_0)$, and therefore $P_{(G,R)}$ is a face of $P_{(G_0,R_0)}$ that is obtained by setting to zero all variables corresponding to edges in $E(G_0)\setminus E(G)$.
Because a description of $P_{(G,R)}$ can be obtained from that of $P_{(G_0,R_0)}$ by simply removing all variables corresponding to edges in $E(G_0)\setminus E(G)$, the complexity of a description of $P_{(G,R)}$ in the sense of \cref{thm:diverse_and_large_coeff} carries over to $P_{(G_0,R_0)}$.
Nonetheless, we stick to considering concrete graphs $(G,R)$ for the simplicity of the presentation.

\medskip

We show that we can carry over our facet complexity result to bimodular integer programs by exploiting a natural bimodular integer programming formulation for bipartite red-odd perfect matchings.
In this transfer, one has to take care of the fact that the natural bimodular formulation lives in extended space, whereas the statement we want to obtain is about the original space.
\begin{theorem}\label{thm:bimodular_intro}
    For every $m\in \mathbb{Z}_{>0}$, there is an $m$-variable bimodular integer program whose associated polytope has facets that, in any inequality description, require $\Omega(m^{1/4})$ different coefficients.
\end{theorem}
Previously, it was known that bimodular programs capture problems with exponential extension complexity (see, for example,~\cite{cevallosLiftingLinearExtension2018,jia_2023_exactMatching}).
However, prior to our work, the potential existence of clean inequality descriptions with small coefficients was not ruled out.

\medskip

As a side result, we show that the candidate description $Q_{(G,R)}$ of \citeauthor{jia_2023_exactMatching} does not admit efficient separation (unless $\P = \NP$).

\begin{theorem}\label{thm:separation}
    Let $G=(V,E)$ be a bipartite graph with red edges $R\subseteq E$.
    Given a point $x\in\mathbb{R}_{\geq 0}^E$, it is \NP-complete to decide whether $x\notin Q_{(G,R)}$.
    The problem remains \NP-complete if $x$ is known to be in the perfect matching polytope of $G$.
\end{theorem}
This is in interesting contrast to the fact that the equivalence of optimization and separation implies that polynomial-time separation over $P_{(G,R)}$ is possible, and may arguably be a reason that would make it surprising (though not impossible) to see this constraint family play a major role in an exact description of $P_{(G,R)}$.

\subsection{Organization of this paper}
\cref{sec:approach} is dedicated to giving an overview of our approach and techniques to prove \cref{thm:diverse_and_large_coeff}, with some of the technical proofs deferred to \cref{sec:complex-facets}.
\Cref{thm:bimodular_intro} is shown in \cref{sec:bimodular_polytope}.
Moreover, a proof of \cref{thm:separation} is provided in \cref{sec:separation}, and a concrete counterexample showing $Q_{(G,R)} \neq P_{(G,R)}$ is given in \cref{sec:explicit-example}.

\section{An overview of our approach towards complex facets}
\label{sec:approach}

We now give an overview of our approach for proving \cref{thm:diverse_and_large_coeff}.
Concretely, we will explicitly construct a family of facets of $P_{(G,R)}$ with the properties claimed in the theorem.
To this end, we exploit a link between the odd-red bipartite perfect matching polytope and the odd cycle polytope.
For a (not necessarily bipartite) graph $G=(V,E)$, the \emph{odd cycle polytope} is defined as
\[
	P_{\textup{odd}}(G) = \mathrm{conv}\{\chi^C\colon C\subseteq E \text{ is an odd cycle in G}\}\enspace.
\]

To obtain a correspondence between odd cycles and odd-red matchings, we use the following construction:
Starting from a graph $G=(V,E)$, we consider a bipartite graph $\widehat G = (\widehat V, \widehat E)$, where $\widehat V$ consists of two copies $V^+\coloneqq \{v^+\colon v\in V\}$ and $V^-\coloneqq \{v^-\colon v\in V\}$ of $V$, and $\widehat E$ contains all edges $\{v^+, v^-\}$ for $v\in V$, as well as all edges $\{u^+, v^-\}$ and $\{u^-, v^+\}$ for every edge $\{u,v\}\in E$.
Moreover, we define $\widehat R\coloneqq \widehat E \setminus \{\{v^+, v^-\}\colon v\in V\}$ to be the subset of red edges in $\widehat G$.
See \Cref{fig:cycle-to-matching} for an illustration of this construction.

\begin{figure}[ht]
	\begin{center}
\pgfdeclarelayer{bg}    %
\pgfsetlayers{bg,main}  %

\begin{tikzpicture}[scale=0.4]

	\coordinate (3) at (6.60,3.0) {};
	\coordinate (4) at (3.60,6.0) {};
	\coordinate (1) at (0.60,3.0) {};
	\coordinate (2) at (3.60,0.0) {};

	\begin{scope}[every node/.style={node_style}]
		\node at (1) {};
		\node at (2) {};
		\node at (3) {};
		\node at (4) {};
	\end{scope}

	\begin{scope}
		\begin{scope}
			\node[right=1pt] at (3) {$c$};
			\node[above=1pt] at (4) {$d$};
			\node[left=1pt] at (1) {$a$};
			\node[below=1pt] at (2) {$b$};
		\end{scope}

		\begin{pgfonlayer}{bg}
			\begin{scope}[
					every path/.style={
						standard_line,
					},
				]
				\draw[highlight] (1) -- (2);
				\draw (1) -- (2);
				\draw (2) -- (3);
				\draw[highlight] (2) -- (4);
				\draw (2) -- (4);
				\draw (3) -- (4);
				\draw[highlight] (4) -- (1);
				\draw (4) -- (1);
			\end{scope}
		\end{pgfonlayer}
	\end{scope}

	\begin{scope}[xshift=15cm]
		\coordinate (11) at (0,2) {};
		\coordinate (21) at (0,0) {};
		\coordinate (31) at (0,6) {};
		\coordinate (41) at (0,4) {};
		\coordinate (12) at (5,2) {};
		\coordinate (22) at (5,0) {};
		\coordinate (32) at (5,6) {};
		\coordinate (42) at (5,4) {};
		\begin{scope}[every node/.style={node_style}]
			\node at (11) {};
			\node at (21) {};
			\node at (31) {};
			\node at (41) {};
			\node at (12) {};
			\node at (22) {};
			\node at (32) {};
			\node at (42) {};
		\end{scope}

		\begin{scope}
			\node[left=1pt] at (11) {$c^+$};
			\node[left=1pt] at (21) {$d^+$};
			\node[left=1pt] at (31) {$a^+$};
			\node[left=1pt] at (41) {$b^+$};
			\node[right=1pt] at (12) {$c^-$};
			\node[right=1pt] at (22) {$d^-$};
			\node[right=1pt] at (32) {$a^-$};
			\node[right=1pt] at (42) {$b^-$};
		\end{scope}

		\begin{pgfonlayer}{bg}
			\begin{scope}[
					every path/.style={
						standard_line,
						black
					},
				]
				\draw (21) -- (22);
				\draw (31) -- (32);
				\draw (41) -- (42);
				\draw[highlight] (11) -- (12);
				\draw (11) -- (12);
			\end{scope}
			\begin{scope}[
					every path/.style={
						standard_line,
						red
					},
				]
				\draw (41) -- (12);
				\draw (12) -- (21);
				\draw (22) -- (31);
				\draw (32) -- (41);
				\draw (42) -- (11);
				\draw (42) -- (21);
				\draw (11) -- (22);
				\draw[highlight] (21) -- (32);
				\draw (21) -- (32);
				\draw[highlight] (31) -- (42);
				\draw (31) -- (42);
				\draw[highlight] (41) -- (22);
				\draw (41) -- (22);
			\end{scope}
		\end{pgfonlayer}
	\end{scope}
\end{tikzpicture}
 	\end{center}
	\caption{An initial graph $G$ (left), and the corresponding bipartite graph $\widehat G$ with edges of $\widehat R$ indicated in red (right).
	It is highlighted (in gray) how an odd cycle in $G$ gets mapped to an odd-red perfect matching in $\widehat G$.}
	\label{fig:cycle-to-matching}
\end{figure}
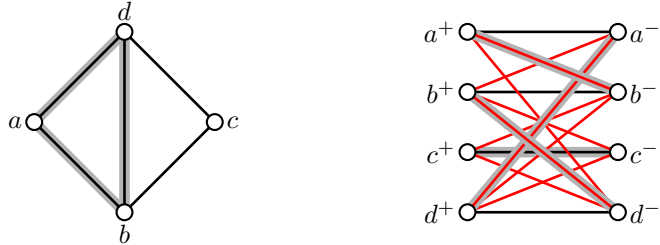

Now an odd cycle $C$ on vertices $v_1, \ldots, v_k, v_{k+1}=v_1$ in $G$ can be translated, for example, to the odd-red perfect matching $M$ in $\widehat G$ with edges $\{v_i^+, v_{i+1}^-\}$ for $i\in [k]$, and $\{v^+, v^-\}$ for all $v\notin C$.
Note that we are not claiming a one-to-one correspondence between odd cycles and odd-red perfect matchings here: one could also construct another odd-red perfect matching from the same odd cycle by traversing it in the opposite direction, thus using the edges $\{v_i^-, v_{i+1}^+\}$.
Conversely, by reversing the construction above, one can see that in general, an odd-red perfect matching $M$ in $\widehat G$ can be mapped to a disjoint union of cycles in $G$ for which the total number of edges is odd (some of the cycles may be degenerate cycles that consist of one edge used twice).
Crucially, such a union of cycles will always contain at least one odd cycle, and thus the corresponding incidence vector is contained in the \emph{dominant} of the odd cycle polytope, which is defined as the Minkowski sum
\[
	\domP(G) = P_{\textup{odd}}(G) + \mathbb{R}_{\geq 0}^E\enspace.
\]
We will in the following identify a family of complex facets of $\domP(G)$ for complete graphs $G$, and then transfer these to facets of $P_{(\widehat{G}, \widehat{R})}$.
The main ingredients for this transformation are laid out next.

\subsection[Complex facets of \texorpdfstring{$\domP(G)$}{the dominant of the odd cycle polytope}]{\boldmath Complex facets of \texorpdfstring{$\domP(G)$}{the dominant of the odd cycle polytope}\unboldmath}%

We begin by showing that $\domP(G)$ belongs to the interesting family of polyhedra arising from combinatorial optimization problems that allow for efficient optimization algorithms, but have complex facets.
The result we obtain in terms of facet complexity is even stronger than in \cref{thm:diverse_and_large_coeff}, as we can show the need for a linear number of different coefficients.

\begin{theorem}\label{thm:diverse_and_large_coeff_odd_cycle}
    For every odd $n\in \mathbb{Z}_{>0}$, there exists a graph $G = (V,E)$ on $n$ vertices and an exponentially (in $n$) sized family $\mathcal{F}$ of facets of the polytope $\domP(G)$ with the following property:
    For every $F\in\mathcal{F}$ and all $a\in \mathbb{Z}^E$, $b\in \mathbb{Z}$ such that $F = \{x\in \domP(G)\colon a^\top x=b\}$, we have
    \[
      \max_{e\in E} |a_e| \geq n - 4
      \qquad\text{and}\qquad
      \lvert\{a_e\colon e\in E\}\rvert\geq \frac{n-3}{2}\enspace.
    \]
\end{theorem}

To prove \cref{thm:diverse_and_large_coeff_odd_cycle}, it is convenient to parameterize the odd size $n$ of the graph by $n = 2k+3$.
Let $G\coloneqq K_{n}$ be the complete graph on $n$ vertices with vertex set $V$ and edge set $E$.
We explicitly construct a family of facets $\mathcal{F}$ induced by $(n-2)$-cycles in $G$, i.e., cycles of length $n-2 = 2k+1$ in $G$.
Concretely, consider a $(2k+1)$-cycle $C = (v_1, \dots, v_{2k+1})$ in $G$, and let $s$ and $t$ be the two vertices of $G$ not contained in $C$.
For an edge $\{v_i, v_j\}$ between two vertices of $C$, set
\begin{equation*}
	\ell(\{v_i, v_j\}) = \begin{cases}
		|j-i| & \text{if } |j-i| \text{ is odd}\\
		2k+1 - |j-i| & \text{otherwise}
	\end{cases}\enspace.
\end{equation*}
Note that equivalently, $\ell(\{v_i, v_j\})$ can be defined as the length of the odd-length path on $C$ from $v_i$ to $v_j$ (also see \cref{fig:definition-ell}).

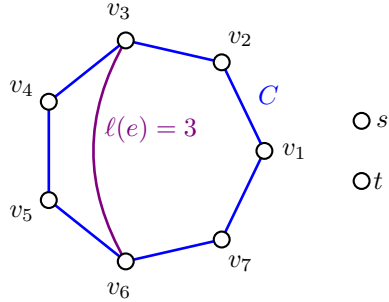
\begin{figure}[ht]
	\centering

	\begin{tikzpicture}[scale=0.4]

		\pgfdeclarelayer{background}
		\pgfsetlayers{background,main}

		\begin{scope}[every
			node/.style={node_style
      }]
			\node (v) at (7,1) {};
			\node (w) at (7,-1) {};
		\end{scope}

		\node[right=2pt] at (v) {$s$};
		\node[right=1pt] at (w) {$t$};

		\node[regular polygon, regular polygon sides=7, minimum size=3cm, rotate=-90] (a) at (0,0) {};

		\begin{scope} [
				every path/.style={
					standard_line,
				},
				]
			\begin{scope} [blue, line width=1pt]
		  	\draw (a.corner 3) -- (a.corner 2) -- (a.corner 1) -- (a.corner 7) -- (a.corner 6);
			\end{scope}
			\begin{scope} [blue, line width=1pt]
		  	\draw (a.corner 3) -- (a.corner 4) -- (a.corner 5) -- (a.corner 6);
			\end{scope}

			\begin{scope}[violet, line width=1pt]
				\draw (a.corner 3) to[bend right] node[pos=0.4, right] {$\ell(e)=3$} (a.corner 6);
			\end{scope}
		\end{scope}

		\foreach \i in {1,2,...,7} {
			\node[node_style] (\i) at (a.corner \i) {};
		}

		\node[regular polygon, regular polygon sides=7, minimum size=3.8cm, rotate=-90] (labels) at (0,0) {};
		\foreach \i in {1,2,...,7} {
			\node (\i) at (labels.corner \i) {$v_\i$};
		}

		\node[above right=-10pt, blue] at ($(1)!0.5!(2)$) {$C$};
	\end{tikzpicture}
 	\caption{Partial visualization of the graph $G$ for $k=3$ ($n=9$). The edge $e=\{v_3, v_6\}$ has $\ell$-value $3$ because the odd-length path in the cycle $C$ connecting its endpoints is $(v_3, v_4, v_5, v_6)$ of length $3$.}
		\label{fig:definition-ell}
\end{figure}

With this notation at hand, we can define the following family of constraints, where $E[S]$, for $S\subseteq V$, denotes the set of all edges with both endpoints in $S$.

\begin{definition}[\boldmath$C$\unboldmath-induced constraint]
  Let $G = K_n$ for some odd $n\in \mathbb{Z}_{\geq 3}$, and $k=\frac{n-3}{2}$.
  Consider an $(n-2)$-cycle $C$ in $G$, and let $s$ and $t$ be the two vertices of $G$ not contained in $C$. 
  We define the \emph{$C$-induced constraint} for $\domP(G)$ by
    \begin{equation*}
	    x(\{s,t\}) + k \cdot x(\delta(C)) + \sum_{e\in E[V\setminus \{s,t\}]} \ell(e) \cdot x(e) \geq 2k+1 \enspace,
    \end{equation*}
    Moreover, denoting by $\mathcal{C}$ the set of all $(n-2)$-cycles in $G$, we refer to the set of all $C$-induced constraints for $C\in \mathcal{C}$ as the \emph{$\mathcal{C}$-induced constraints} for $\domP(G)$.
\end{definition}

\cref{fig:clique-constraint} shows the coefficients of a $C$-induced constraint of $\domP(K_9)$.
Note that a $C$-induced constraint indeed only depends on the choice of a cycle $C$ of length $n-2$ in $G$, and two different such cycles lead to different constraints.
In particular, the $\mathcal{C}$-induced constraints form an exponential family of constraints.
Moreover, for $n = 2k + 3 \geq 5$, the coefficients of a $C$-induced constraint are $\{1, 3, 5, \ldots, 2k-1\} \cup \{k\}$, which are at least $k=\frac{n-3}{2}$ distinct values.

The validity of a $C$-induced constraint can be verified as follows.
Let $x= \chi^{\overline{C}}$ be the characteristic vector of an odd cycle $\overline{C}$ in $G$.
If $\overline{C}$ goes through $s$ or $t$, then the left-hand side term $k \cdot x(\delta(C))$ already contributes $2 k$ units, and one more unit is obtained because there must be a third edge in $\overline{C}$, and all edges have strictly positive coefficients.
Otherwise, if $\overline{C}$ does not go through $s$ or $t$, then we can map each edge $e\in \overline{C}$ to the odd-length path $P_e$ in $C$ connecting its endpoints.
If the paths $P_e$ for $e\in \overline{C}$ cover all of $C$, the constraint is fulfilled because the sum of the $\ell$-values of the edges in $\overline{C}$, which is part of the left-hand side of the $C$-induced constraint, is equal to $\sum_{e\in \overline{E}}|P_e|$, which is at least $2k+1$ because the full cycle is covered.
A key observation, which finishes the reasoning, is that it is impossible that the paths $P_e$ do not cover all of $C$.
If this happened, then the paths $P_e$ for $e\in \overline{C}$ would correspond to an odd number of odd length intervals on a line, as we may break the cycle open at the uncovered edge without breaking up any of the paths $P_e$.
However, it is impossible that such intervals form a closed walk, because for that, their total length would have to be even.

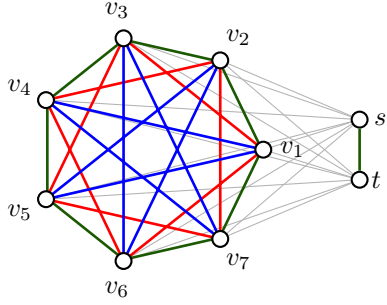
\begin{figure}[!th]
	\centering
	\begin{tikzpicture}[scale=0.4]

		\pgfdeclarelayer{background}
		\pgfsetlayers{background,main}

		\begin{scope}[every
			node/.style={node_style
      }]
			\node (s) at (7,1) {};
			\node (t) at (7,-1) {};
		\end{scope}

		\node[right=2pt] at (s) {$s$};
		\node[right=1pt] at (t) {$t$};

		\begin{scope}[
				every path/.style={
					standard_line,
					darkgreen
				},
			]
		\draw[darkgreen, line width=1pt] (s) -- (t);

		\node[draw, regular polygon, regular polygon sides=7, minimum size=3cm, rotate=-90] (a) at (0,0) {};
	\end{scope}

		\begin{scope}[
				every path/.style={
					standard_line,
					red
				},
			]
			\draw (a.corner 1) -- (a.corner 3) -- (a.corner 5) -- (a.corner
			7) -- (a.corner 2) -- (a.corner 4) --(a.corner 6) -- (a.corner 1);
		\end{scope}

		\begin{scope}[
				every path/.style={
					standard_line,
					blue
				},
			]
			\draw (a.corner 1) -- (a.corner 4) -- (a.corner 7) -- (a.corner
			3) -- (a.corner 6) -- (a.corner 2) --(a.corner 5) -- (a.corner 1);
		\end{scope}

		\foreach \i in {1,2,...,7} {
			\begin{pgfonlayer}{background}
				\draw[thin, black!30] (s) -- (a.corner \i);
				\draw[thin, black!30] (t) -- (a.corner \i);
			\end{pgfonlayer}
			\node[node_style] (\i) at (a.corner \i) {};
		}

		\node[regular polygon, regular polygon sides=7, minimum size=3.8cm, rotate=-90] (labels) at (0,0) {};
		\foreach \i in {1,2,...,7} {
			\node (\i) at (labels.corner \i) {$v_\i$};
		}
	\end{tikzpicture}
 	\caption{Visualization of a $C$-induced-constraint of $\domP(K_9)$.
		The coefficient of green edges is $1$, of blue edges $3$, and of red edges $5$.
		Furthermore, the edges between the $7$-gon and the vertices $s$ and $t$, marked in light gray, also have a coefficient of $3$.
        The right-hand side of the constraint is $7$.}
		\label{fig:clique-constraint}
\end{figure}

The intuition behind the choice of the coefficients for a $C$-induced constraint is---besides the goal of constructing a valid constraint---to ensure that many odd-length cycles are tight for the constraint.
The definition of $\ell$ ensures that the odd-length cycle defined by any chord $\{v_i,v_j\}$ of $C$ together with the even part of $C$ connecting $v_i$ to $v_j$ is tight for the constraint.
Additionally, all triangles that contain an edge of $C$ or the edge $\{s,t\}$, together with two appropriate edges crossing the cut $\delta(C)$, are tight for the constraint.
Showing that enough of these tight cycles (respectively, their incidence vectors) are linearly independent will allow for showing that a $C$-induced constraint is indeed facet-defining.

\begin{lemma}\label{lem:c-induced-facet}
    $\mathcal{C}$-induced constraints are facet-defining for $\domP(G)$.
\end{lemma}

A proof of \cref{lem:c-induced-facet} is 
provided in \cref{sec:complex-facets}.
Note that, together with the explicit description of $\mathcal{C}$-induced constraints and full-dimensionality of $\domP(G)$, \cref{lem:c-induced-facet} immediately implies \cref{thm:diverse_and_large_coeff_odd_cycle}.
Indeed, $\domP(G)$ being full dimensional implies that, given a facet, any inequality describing it uses the same normal vector up to scaling.
Our goal is to transfer these complex facets of $\domP(G)$ to facets of $P_{(\widehat G, \widehat R)}$, while showing that much of their complexity is preserved.
We detail how this can be achieved next.

\subsection[Transferring facets of \texorpdfstring{$\domP(G)$}{the dominant of the odd cycle polytope} to \texorpdfstring{$P_{(\widehat G, \widehat R)}$}{the odd-red perfect matching polytope}]{\boldmath Transferring facets of \texorpdfstring{$\domP(G)$}{the dominant of the odd cycle polytope} to \texorpdfstring{$P_{(\widehat G, \widehat R)}$}{the odd-red perfect matching polytope}\unboldmath}

Because $\domP(G)$ is a dominant and therefore up-closed, it can be described by so-called \emph{covering inequalities}, i.e., inequalities of the form
\begin{align}\label{eq:facet-odd-cycle-dom}
	\sum_{e\in E} a_e x_e \geq b
\end{align}
with $a \in \mathbb{R}^E_{\geq 0}$ and $b\in \mathbb{R}$.
For any such constraint, we define its \emph{canonical transformation} as
\begin{align}\label{eq:facet-red-odd-matching}
	\sum_{\{u, v\} \in E} a_{\{u,v\}}\left(y_{\{u^+, v^-\}} + y_{\{u^-, v^+\}}\right) \geq b\enspace,
\end{align}
where $y$ is defined on the edges of $\widehat{G}$.
The following is an immediate observation.

\begin{observation}\label{obs:feasible}
	The canonical transformation of any constraint that is valid for $\domP(G)$ is valid for $P_{(\widehat G, \widehat R)}$.
\end{observation}

\begin{proof}
	Consider a valid constraint $\kappa$ for $\domP(G)$ of the form given in \eqref{eq:facet-odd-cycle-dom}, and its canonical transformation as given in \eqref{eq:facet-red-odd-matching}.
    Let $y\in P_{(\widehat G, \widehat R)}$ be the incidence vector of an odd-red perfect matching $M$ in $\widehat G$.
    As discussed earlier, $M$ can be mapped to a union of cycles in $G$ that contains at least one odd cycle $C$, by mapping edges $\{v^\pm, w^\mp\}$ in $M$ back to the edge $\{v,w\}$ in $G$.
    Let $x\in \domP(G)$ be the incidence vector of $C$.
    Then, the left-hand side of the canonical transformation of $\kappa$ evaluated at $y$ is at least the left-hand side of $\kappa$ evaluated at $x$.
    Consequently, validity of $\kappa$ for $\domP(G)$ implies validity of its canonical transformation for~$P_{(\widehat G, \widehat R)}$.
\end{proof}

For $\mathcal{C}$-induced constraints, which are facets of $\domP(G)$ by \cref{lem:c-induced-facet}, we claim that the canonical transformation not only gives feasible constraints as shown in \cref{obs:feasible} above, but facet-defining constraints.
In order to show that a constraint is facet-defining, we need to identify enough independent points on the facet.
Naturally, for a facet $\kappa$ of $\domP(G)$, the hope is to be able to take independent (characteristic vectors of) odd cycles for which $\kappa$ is tight, transfer these to odd-red matchings for which the canonical transformation of $\kappa$ is tight, and argue that these matchings are independent enough to conclude that the canonical transformation of $\kappa$ is facet-defining for $P_{(\widehat G, \widehat R)}$.
For carrying out this final argument of independence on the odd-red matching side in an almost black-box way, we exploit the following notion of \emph{expressibility} on the side of odd cycles.

\begin{definition}[Expressibility]\label{def:expressibility}
    Let $G=(V,E)$ be a graph and consider a constraint $\kappa$ of the form given in \eqref{eq:facet-odd-cycle-dom} that is valid for $\domP(G)$.
    Let 
    $$
    \mathcal{T} \coloneqq \left\{ C \subseteq E \colon C\text{ is an odd cycle in }G \text{ and } \kappa\text{ is tight for }\chi^C\right\}
    $$
    be the set of odd cycles in $G$ whose incidence vectors satisfy $\kappa$ with equality.
    \begin{enumerate}
        \item For $Q\subseteq E$, we define the $Q$-expressible edges $E_Q\subseteq E$ through the following iterative procedure: We start with $E_Q = Q$, and as long as there is an edge $e\in E\setminus E_Q$ and a cycle $C\in \mathcal{T}$ with $e\in C$ and $C\setminus\{e\}\subseteq E_Q$, we add $e$ to~$E_Q$.
        \item We say that $\kappa$ is \emph{$k$-expressing} if there is a set $Q\subseteq E$ of size at most $k$ such $E_Q=E$.
    \end{enumerate}
\end{definition}

In other words, a constraint is $k$-expressing if it suffices to provide $k$ edges through which all others can be expressed using odd cycles as defined above.
As one might imagine, this number $k$ is closely related to an upper bound on the co-dimension of the resulting face when a $k$-expressing constraint of $\domP(G)$ is transformed to a constraint of $P_{(\widehat G, \widehat R)}$.
For small enough $k$, exploiting that we know the dimension of $P_{(\widehat G, \widehat R)}$, we will hence be able to certify that the resulting constraint is facet-defining.
Concretely, our result is the following.
(See %
\cref{sec:complex-facets}
for proofs of the following statements.)

\begin{lemma}\label{lem:expressing-facet-transfer}
    Let $G$ be an $n$-vertex graph, and let $\kappa$ be an $(n-1)$-expressing facet-defining constraint of $\domP(G)$ that is not tight for all of $P_{\textup{odd}}(G)$.
    Then, the canonical transformation of $\kappa$ is facet-defining for $P_{(\widehat G, \widehat R)}$.
\end{lemma}

With \cref{lem:expressing-facet-transfer} at hand, using that for $K_n$ no facet is tight for all odd cycles, all that remains to show is that $\mathcal{C}$-induced constraints of $\domP(G)$ are $(n-1)$-expressing.

\begin{lemma}\label{lem:n-1expressing}
    $\mathcal{C}$-induced constraints are $(n-1)$-expressing.
\end{lemma}

Together, \cref{lem:expressing-facet-transfer,lem:n-1expressing} immediately yield the following corollary that provides the desired facets of $P_{(\widehat G, \widehat R)}$.

\begin{corollary}
	The canonical transformation of any $\mathcal{C}$-induced constraint is facet-defining for $P_{(\widehat G, \widehat R)}$.
\end{corollary}

We refer to the facets of $P_{(\widehat G, \widehat R)}$ obtained in this way as \emph{transformed $\mathcal{C}$-induced facets} of $P_{(\widehat G, \widehat R)}$.

\subsection[Asserting complexity of transformed \texorpdfstring{$\mathcal{C}$}{C}-induced facets]{Asserting complexity of transformed \texorpdfstring{\boldmath$\mathcal{C}$\unboldmath}{C}-induced facets}
\label{sec:complexity-analysis}

Note that transformed $\mathcal{C}$-induced facets of $\domP(G)$ inherit the complexity of the original facets in terms of coefficients in the linear inequality representations, as these remain unchanged through the transformation.
However, while facets of $\domP(G)$ have uniquely defined normal vectors (up to scaling) because $\domP(G)$ is full-dimensional, this is not the case for $P_{(\widehat G, \widehat R)}$ because $$\dim(P_{(\widehat G, \widehat R)}) = |\widehat E| - |\widehat V| + 1\enspace.$$
In fact, the minimal subspace containing $P_{(\widehat G, \widehat R)}$ is determined by the degree constraints $y(\delta(v)) = 1$ for all $v\in \widehat V$.
In other words, any facet-defining constraint $a^\top y \geq b$ of $P_{(\widehat G, \widehat R)}$ can be equivalently expressed in the form
\begin{equation}\label{eq:equivalent-facet-rep}
    \mu \cdot a^\top y + \sum_{v\in \widehat V} \lambda_v \cdot y(\delta(v)) \geq \mu \cdot b + \sum_{v\in \widehat V} \lambda_v
\end{equation}
for arbitrary $\mu > 0$ and $\lambda \in \mathbb{R}^{\widehat V}$.

Hence, to make sure that a facet of $P_{(\widehat G, \widehat R)}$ admits no low complexity representation, we need to show than every such equivalent representation has high complexity. 
(Note that we typically assume that facets are scaled so that they have integral coefficients and integral right-hand side.)

To this end, we change our viewpoint and treat coefficient vectors $a\in\mathbb{Z}^{\widehat E}$ as matrices in $\mathbb{Z}^{V^+ \times V^-}$, with rows indexed by $V^+$ and columns indexed by $V^-$.
More precisely, for $u\in V^+$ and $v\in V^-$, the matrix entry $a(u, v)$ shall be the coefficient $a_{\{u, v\}}$ of the variable associated with the edge $\{u, v\}\in \widehat E$.
In \cref{fig:matrix-example}, we exemplify this matrix viewpoint by presenting the matrix corresponding to the transformation of the $C$-induced facet of $\domP(K_9)$ presented earlier in \cref{fig:clique-constraint}.
To see how the matrix representation changes under moving to an equivalent form as in \eqref{eq:equivalent-facet-rep}, which we write as $(a^{\mu,\lambda})^\top y \geq b^{\mu,\lambda}$ for $\mu > 0$ and $\lambda \in \mathbb{R}^{\widehat V}$, observe that for $u\in V^+$ and $v\in V^-$, we have
\begin{equation}\label{eq:matrix-representation-equivalence}
	a^{\mu,\lambda}(u, v) = \mu \cdot a(u, v) + \lambda_{u} + \lambda_{v}\enspace.
\end{equation}
In other words, the matrix $a^{\mu,\lambda}$ can be obtained from the matrix $a$ by scaling with $\mu$, and then adding a value $\lambda_u$ to all entries in row $u$ and a value $\lambda_v$ to all entries in column $v$, for all $u\in V^+$ and $v\in V^-$.
Thus, the analysis of equivalent facet representations boils down to analyzing the impact of these simple matrix operations on the matrix obtained from the coefficient vector $a$.

\begin{figure}[h]
	\centering
	\setlength{\tabcolsep}{0.05cm}
	\begin{tabular}{c|@{\hspace{0.1cm}}ccccccc@{\hspace{0.2cm}}cc}
		& $v_1^-$ & $v_3^-$ & $v_5^-$ & $v_7^-$ & $v_2^-$ & $v_4^-$ & $v_6^-$ & $s^-$ & $t^-$ \\[0.05cm]
		\hline
		\rule{0pt}{1.2em}
		$v_1^+$ & 0 & 5 & 3 & 1 & 1 & 3 & 5 & 3 & 3 \\
		$v_3^+$ & 5 & 0 & 5 & 3 & 1 & 1 & 3 & 3 & 3 \\
		$v_5^+$ & 3 & 5 & 0 & 5 & 3 & 1 & 1 & 3 & 3 \\
		$v_7^+$ & 1 & 3 & 5 & 0 & 5 & 3 & 1 & 3 & 3 \\
		$v_2^+$ & 1 & 1 & 3 & 5 & 0 & 5 & 3 & 3 & 3 \\
		$v_4^+$ & 3 & 1 & 1 & 3 & 5 & 0 & 5 & 3 & 3 \\
		$v_6^+$ & 5 & 3 & 1 & 1 & 3 & 5 & 0 & 3 & 3 \\[0.2cm]
		$s^+$   & 3 & 3 & 3 & 3 & 3 & 3 & 3 & 0 & 1 \\
		$t^+$   & 3 & 3 & 3 & 3 & 3 & 3 & 3 & 1 & 0 \\
	\end{tabular}
	\caption{A matrix representation of the coefficient vector of the translated $C$-induced facet of $\domP(K_9)$ presented earlier in \cref{fig:clique-constraint}.}
	\label{fig:matrix-example}
\end{figure}

Our complexity results focus on lower bounds on the largest entry in absolute value and the number of different entries that can be achieved through the matrix operations described above.
Using specific invariant alternating sums of matrix entries, we can provide a linear lower bound on the first quantity.
Interestingly, for the number of different entries, an analysis based on the pigeonhole principle that takes into account only two rows of the matrix already results in a lower bound of order $\Omega(\sqrt{n})$, completing  the proof of \cref{thm:diverse_and_large_coeff} (see %
\cref{sec:complex-facets}
for details).

While the precise asymptotic behavior of the smallest achievable number of different coefficients in inequality representations of transformed $\mathcal{C}$-induced facets is left open, we present a way to reduce to $O(n^{\sfrac{2}{3}})$ many different coefficients (again, we refer to %
\cref{sec:complex-facets}
for details).

\section{Bimodular polytopes with complex facets}%
\label{sec:bimodular_polytope}

In this section, we prove \cref{thm:bimodular_intro} by showing that a natural bimodular formulation of the red-odd perfect matching problem admits the postulated facet complexity.
To this end, let $G=(V,E)$ be a bipartite graph with red edges $R\subseteq E$, and observe that the integral solutions of the system
\begin{align}\label{eq:bimodular_system}
\begin{array}{rclr}
  x(\delta(v)) & = & 1 & \forall v\in V \\
  x(R) - 2y & = & 1 &\\ 
  x & \in & \mathbb{R}^E_{\geq 0} \\
  y & \in & \mathbb{R}_{\geq 0} 
\end{array}
\end{align}
are precisely the pairs $(\chi^M, r)$, where $M$ is a red-odd perfect matching in $G$ and $r = \frac{|M\cap R| - 1}{2}$.
Also, the system \eqref{eq:bimodular_system} is indeed bimodular: one can observe that all full-rank subdeterminants of the constraint matrix are in $\{-2,0,2\}$.

Let $P_{\textup{bimod}}$ be the convex hull of all integral solutions to \eqref{eq:bimodular_system}.
By the above observation, we have
\[
P_{\textup{bimod}} = \left\{ \left(x,\frac{x(R)-1}{2}\right) : x\in P_{(G,R)} \right\}\enspace,
\]
hence $P_{\textup{bimod}}$ is obtained from $P_{(G,R)}$ by lifting along the new variable $y$ to the hyperplane defined by $y = \frac{x(R)-1}{2}$.
Thus, we get the following.

\begin{observation}\label{obs:facet_translation}
For every facet $F$ of $P_{\textup{bimod}}$, there is a facet $F'$ of $P_{(G,R)}$ such that
\[
  F = \left\{ \left(x,\frac{x(R)-1}{2}\right) : x\in F'\right\}\enspace,
\]
and vice versa. In particular, for all $a\in\mathbb{R}^E$ and $b,c\in\mathbb{R}$, we have $F=\{x\in P_{\textup{bimod}}\colon a^\top x + c y = b\}$ if and only if
\[
  F' = \left\{x\in P_{(G,R)}\colon\left(a+\frac{c}{2}\chi^R\right)^\top x = b + \frac{c}{2}\right\}\enspace.
\]
\end{observation}

This observation has the following two consequences.
First, it allows to directly lift any complex facet $F'$ of $P_{(G,R)}$ to a facet $F$ of $P_{\textup{bimod}}$.
Second, it implies that any representation of such a lifted facet $F$ through a linear equality can be translated to a representation of the original facet $F'$.
Because the transformation detailed in \cref{obs:facet_translation} can at most double the number of different coefficients, and increase the largest absolute value by a factor of at most $\frac32$, \cref{thm:diverse_and_large_coeff} directly implies the following strengthening of \cref{thm:bimodular_intro}.

\begin{corollary}\label{thm:diverse_and_large_coeff_bimodular}
  For every even $n\in \mathbb{Z}_{>0}$, there exists a bipartite graph $G = (V,E)$ on $n$ vertices, red edges $R\subseteq E$, and an exponentially (in $n$) sized family $\mathcal{F}$ of facets of the polytope $P_{\textup{bimod}}$ with the following property:
  For every $F\in\mathcal{F}$ and all $a\in \mathbb{Z}^E$, $b\in \mathbb{Z}$ such that $F = \{x\in P_{\textup{bimod}}\colon a^\top x=b\}$, we have
  \[
    \max_{e\in E} |a_e| \geq \frac{n-4}{3}
    \qquad\text{and}\qquad
    \lvert\{a_e\colon e\in E\}\rvert\geq \sqrt{\frac{n-1}{8}}\enspace.
  \]
\end{corollary}

Note that the number of variables in the bimodular formulation of a red-odd perfect matching problem needs one variables per edge of the bipartite graph plus one additional variable (for $r$).
Hence, a dense $n$-vertex bipartite graph has $m=\Theta(n^2)$ edges, and the number of different coefficients in \cref{thm:diverse_and_large_coeff_bimodular} is in $\Omega(\sqrt{n}) = \Omega(m^{1/4})$, as stated in \cref{thm:bimodular_intro}.

\begingroup
\sloppy
\printbibliography
\endgroup

\appendix

\section[Missing proofs from Section~\ref{sec:approach}]{Missing proofs from \cref{sec:approach}}\label{sec:complex-facets}

\subsection[Facets of \texorpdfstring{$\domP(K_n)$}{the dominant of the odd cycle polytope of the complete graph on n vertices} with complex coefficients]{\boldmath Facets of \texorpdfstring{$\domP(K_n)$}{the dominant of the odd cycle polytope of the complete graph on n vertices} with complex coefficients\unboldmath}

We start by providing a detailed proof of \cref{lem:c-induced-facet}, namely that $\mathcal{C}$-induced constraints are facet-defining for $\domP(K_n)$.
As already pointed out in \cref{sec:approach}, the idea is to identify many odd cycles that are tight for a given $C$-induced constraint, and then show that the incidence vectors of these cycles linearly span the entire space $\mathbb{R}^E$.

\begin{proof}[Proof of \cref{lem:c-induced-facet}]
	Let $n=2k+3$ for some $k\in\mathbb{Z}_{\geq 0}$, let $C = (v_1, v_2, \dots, v_{2k+1})$ be an odd cycle in $G = K_{2k+3}$, and let $s, t$ be the two vertices in $V\setminus V(C)$.
	We show that the $C$-induced constraint 
	\begin{equation*}
		x(\{s,t\}) + k \cdot x(\delta(C)) + \sum_{e\in E[V\setminus \{s,t\}]} \ell(e) \cdot x(e) \geq 2k+1
	\end{equation*}
	is facet-defining for $\domP(G)$.
	
	First, we show that the constraint is valid for $\domP(G)$.
	We start with the observation that every variable $x_e$, for $e\in E$, has a strictly positive coefficient in the $C$-induced constraint.
	Thus, it suffices to show that the constraint is satisfied by the characteristic vector $\chi^D$ for every odd cycle $D$ in $G$.
	We distinguish two cases as follows.
	\begin{itemize}
	\item \emph{$D$ contains an edge of $\delta(C)$.}
	Then $D$ contains at least two edges of $\delta(C)$, plus at least one other edge.
	In this case, the two edges in $\delta(C)$ contribute at least $2k$ to the left-hand side of the $C$-induced constraint. The third edge contributes at least $1$, hence the vector $\chi^D$ satisfies the $C$-induced constraint.

	\item \emph{$D$ does not contain an edge of $\delta(C)$.}
	In this case, $D$ only connects vertices of $C$.
	We transform $D$ into a closed walk $\widebar{D}$ using only edges of $C$ by replacing each edge of the form $\{v_i, v_j\}$ by the odd-length $v_i$-$v_j$ path on $C$.
	By construction, $\widebar{D}$ is thus a concatenation of an odd number of odd-length paths, and hence of odd length itself.
	More precisely, as we replace $\{v_i,v_j\}$ by a path of length $\ell(\{v_i,v_j\})$, which is the coefficient of $x_{\{v_i,v_j\}}$ in the $C$-induced constraint, the length of $\widebar{D}$ is exactly the value of the left-hand side of the $C$-induced constraint evaluated at $x=\chi^D$.
	To show validity of the constraint, we thus need to show that the length of $\widebar{D}$ is at least $2k+1$.
	To this end, it is enough to show that $\widebar{D}$ contains every edge of $C$ at least once, as $C$ has length $2k+1$.

	Assume, for the sake of deriving a contradiction, that $\widebar{D}$ does not contain every edge of $C$, say, by symmetry, $\{v_{2k+1}, v_1\}\notin \widebar{D}$.
	Being a closed walk, $\widebar{D}$ crosses every cut an even number of times.
	For a fixed $i\in [2k]$, consider the cut induced by $S_i = \{v_1,\dots, v_i\}$.
	Then, $\delta_C(S_i) = \{\{v_{2k+1},v_1\}, \{v_i, v_{i+1}\}\}$.
	By assumption, $\widebar{D}$ does not contain $\{v_{2k+1}, v_1\}$, so we conclude that it covers the edge $\{v_i, v_{i+1}\}$ an even number of times.
	As $i$ was arbitrary in $[2k]$, we obtain that the total number of edges in $\widebar{D}$ is even.
	This is a contradiction to the previous observation that $\widebar{D}$ has odd length.
	\end{itemize}

	Next, we show that the $C$-induced constraint is facet-defining by investigating odd cycles that are tight for the constraint.
	Write the constraint as $a^\top x \geq b$ with $a\in \mathbb{Z}_{\geq 0}^E$ and $b = 2k+1$.
	Dominants are full-dimensional by definition, hence \smash{$\dim(\domP(K_{2k+3})) = |E|$}, and it suffices to show that the constraint defines a face $F \coloneqq \{x\in \domP(G)\colon a^\top x = b\}$ of dimension $|E|-1$.
	Let $S \coloneqq \linspan(\{\chi^C\colon C \text{ odd cycle in }G, \chi^C\in F\})$ be the linear span of the incidence vectors of odd cycles that are tight for the $C$-induced constraint.
	The origin $x=0$ is not feasible for the $C$-induced constraint.
	Thus, for the desired $\dim(F) = |E| - 1$ to hold, it is enough to show $S = \mathbb{R}^E$.
	We prove the latter by showing that $\chi^f\in S$ for every edge $f\in E$ by explicitly revealing linear combinations of points in $S$ that yield $\chi^f$, starting from points $\chi^C$ that we know to be in $S$ by definition.

	To this end, let us name some odd cycles that we are about to use explicitly.
	First, for $i\in [2k+1]$, let $T^i$ be the triangle $(s,t,v_i)$, and let $T_s^i$ and $T_t^i$ be the triangles $(s,v_i, v_{i+1})$ and $(t, v_i, v_{i+1})$, respectively.
	For $i,j\in [2k+1]$ with $i\neq j$, define $C_{i,j}$ to be the odd cycle obtained by combining $\{v_i,v_j\}$ with the even length path $P_{i,j}$ from $v_i$ to $v_j$ on $C$.
	A visualization of these definitions can be found in \cref{fig:visualization-T-C}.
	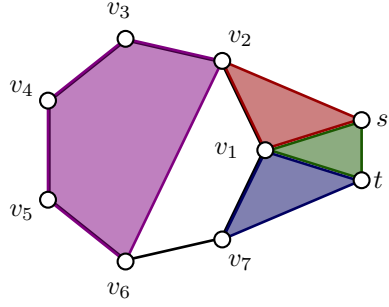
\begin{figure}[ht]
		\centering
	\begin{tikzpicture}[scale=0.4]

		\pgfdeclarelayer{background}
		\pgfsetlayers{background,main}

		\coordinate (v) at (7,1);
		\coordinate (w) at (7,-1);

		\begin{scope}[every	node/.style={node_style}]
			\node (v_node) at (v) {};
			\node (w_node) at (w) {};
		\end{scope}

		\node[right=2pt] at (v) {$s$};
		\node[right=1pt] at (w) {$t$};

		\begin{pgfonlayer}{background}
			\node[draw, line width = 1pt, regular polygon, regular polygon sides=7, minimum size=3cm, rotate=-90] (a) at (0,0) {};

		\newcommand{\shift}{1.3pt}
		\draw[darkgreen, line width=1pt, fill, fill opacity=0.5] (v) -- ($(w)+(0,\shift)$) -- ($(a.corner 1)+(0,\shift)$) -- ($(a.corner 1)-(0,\shift)$) -- ($(v)+(0,-\shift)$);
		\draw[darkred, line width=1pt, fill, fill opacity=0.5] ($(a.corner 1)+(0,\shift)$) -- ($(v)+(0, \shift)$) -- (v) -- (a.corner 2) -- (a.corner 1);
		\draw[darkblue, line width=1pt, fill, fill opacity=0.5] (w) -- (a.corner 7) -- ($(a.corner 1)-(0,\shift)$) -- ($(w)-(0,\shift)$);
		\draw[violet, line width=1pt, fill, fill opacity=0.5] (a.corner 2) -- (a.corner 3) -- (a.corner 4) -- (a.corner 5) -- (a.corner 6) -- (a.corner 2);

		\foreach \i in {1,2,...,7} {
			\node[node_style] (\i) at (a.corner \i) {};
		}

		\node[regular polygon, regular polygon sides=7, minimum size=3.8cm, rotate=-90] (labels) at (0,0) {};
		\foreach \i in {2,...,7} {
			\node (\i) at (labels.corner \i) {$v_\i$};
		}
		\node[xshift=-0.9cm] (1) at (labels.corner 1) {$v_1$};
	\end{pgfonlayer}
	\end{tikzpicture}
 		\caption{
			Cycles $T^i$, $T^i_s$, $T^i_t$, and $C_{i,j}$ for $k=3$.
			$T^1$ is marked in green, $T^1_s$ in red, $T^7_t$ in blue, and $C_{2,6}$ in violet.
		}
		\label{fig:visualization-T-C}
	\end{figure}
	We observe that all these odd cycles correspond to points of $\domP(G)$ that are tight for the $C$-induced constraint:
	\begin{itemize}
		\item For every $i\in [2k+1]$, the triangles $T^i, T_s^i,T_t^i$ all contain two edges of $\delta(C)$ with a coefficient of $k$ and one edge with a coefficient of $1$ in the $C$-induced constraint.
		Thus, $\chi^{T^i}, \chi^{T_s^i},\chi^{T_t^i}\in F$.
		\item For $i,j\in [2k+1]$ with $i\neq j$, we have $a^\top \chi^{C_{i,j}} = |P_{i,j}| + \ell(\{v_i, v_j\}) = |C| = 2k+1$.
		Thus, $\chi^{C_{i,j}}\in F$.
	\end{itemize}
	Additionally, we have $a^\top \chi^C = |C| = 2k+1$, and thus $\chi^C\in F$.
	We are now ready to show that $\chi^f\in S$ for every edge $f\in E$.

	First, consider $f = \{s,t\}$.
	The triangles $\{T^i\}_{i\in[2k+1]}$ cover every edge of $\delta(C)$ once, while the triangles $\{T^i_s\}_{i\in[2k+1]}\cup \{T^i_t\}_{i\in[2k+1]}$ cover each edge of $\delta(C)$ twice.
	To be more precise, we have
	\begin{equation*}
		\sum_{i=1}^{2k+1} \left(2\chi^{T^i} - \chi^{T^i_s} - \chi^{T^i_t} \right)=		(2k+1)\chi^{\{s,t\}} - 2\chi^C \enspace.
	\end{equation*}
	In other words, this expresses  $\chi^{\{s,t\}}$ as a linear combination of points in $S$, hence we conclude $\chi^{\{s,t\}}\in S$.
	Next, consider any edge $\{v_i,v_{i+1}\}$ for $i\in [2k+1]$.
	We have (see \cref{fig:clique-constraint-cycle-edge} for a visualization of the involved triangles)
	\begin{equation*}
		\chi^{T^i_s} + \chi^{T^i_t} - \chi^{T^i} - \chi^{T^{i+1}} =
		2\left(\chi^{\{v_i, v_{i+1}\}} - \chi^{\{s,t\}}\right)\enspace.
	\end{equation*}
	As $\chi^{T^i_s}, \chi^{T^i_t}, \chi^{T^i}, \chi^{T^{i+1}}\in F$ and $\chi^{\{s,t\}}\in S$ by the above, we get $\chi^{\{v_i, v_{i+1}\}}\in S$.
	\begin{figure}[ht]
		\centering

		\begin{tikzpicture}[scale=0.75]

			\pgfdeclarelayer{background}
			\pgfsetlayers{background,main}

			\begin{scope}[every
				node/.style={node_style
			}]
				\node (v1) at (0,1) {};
				\node (v2) at (0,-1) {};

				\node (v) at (2,1) {};
				\node (w) at (2,-1) {};
			\end{scope}

			\node[right=2pt] at (v) {$s$};
			\node[right=1pt] at (w) {$t$};
			\node[left=2pt] at (v1) {$v_1$};
			\node[left=2pt] at (v2) {$v_2$};

			\begin{pgfonlayer}{background}
				\draw[darkgreen, line width=1pt, fill=green] (v) -- (v2) -- (v1) -- (v);
				\draw[darkgreen, line width=1pt] (w) -- (v2) -- (v1) -- (w);
				\fill[fill = darkgreen, opacity=0.5] (w.center) --
				(v2.center) -- (v1.center);
				\fill[fill = darkgreen, opacity=0.5] (v.center) --
				(v2.center) -- (v1.center);
			\end{pgfonlayer}

			\begin{scope}[xshift=5cm]

				\begin{scope}[every
					node/.style={node_style}]
					\node (v1) at (0,1) {};
					\node (v2) at (0,-1) {};

					\node (v) at (2,1) {};
					\node (w) at (2,-1) {};
				\end{scope}
				\node[right=2pt] at (v) {$s$};
				\node[right=1pt] at (w) {$t$};
				\node[left=2pt] at (v1) {$v_1$};
				\node[left=2pt] at (v2) {$v_2$};

				\begin{pgfonlayer}{background}
					\draw[darkred, line width=1pt, fill=darkred, fill opacity=0.5]
					(v.center) -- (w.center) -- (v1.center) -- (v.center);
					\draw[darkred, line width=1pt, fill=darkred, fill opacity=0.5]
					(v.center) -- (w.center) -- (v2.center) -- (v.center);
				\end{pgfonlayer}
			\end{scope}

		\end{tikzpicture}
 		\caption{Adding the two green triangles and subtracting the two red ones gives $\chi^{\{v_1,v_2\}}-\chi^{\{s,t\}}$.}
		\label{fig:clique-constraint-cycle-edge}
	\end{figure}
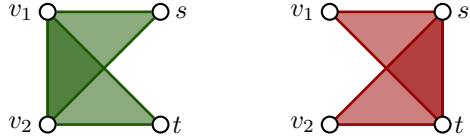

	Next, consider any edge $\{v_i,v_j\}$ for $i,j\in [2k+1]$ with $|i-j|\notin\{0,1\} \pmod*{2k+1}$, i.e., an edge connecting non-consecutive vertices of $C$.
	Then $\chi^{C_{i,j}}$ is in $F$ and $C_{i,j}$ contains, except for the edge $\{v_i,v_j\}$, only edges from $C$.
	For $f\in C$, we already know $\chi^f\in S$, hence we conclude $\chi^{\{v_i,v_j\}}\in S$.
	
	Finally, it remains to show $\chi^f\in S$ for edges $f\in\delta(C)$.
	Without loss of generality, we consider an edge $f=\{v_i,s\}\in\delta(C)\cap\delta(s)$; the proof for edges in $\delta(C)\cap\delta(t)$ is symmetric.
	Observe that we have
	\begin{align*}
			\chi^{T_s^i} - \chi^{T_s^{i+1}} + \ldots - \chi^{T_s^{i-2}} + \chi^{T_s^{i-1}}
		= 2 \chi^{\{v_i,s\}} + \chi^{\{v_i, v_{i+1}\}} - \chi^{\{v_{i+1}, v_{i+2}\}} + \ldots - \chi^{\{v_{i-2}, v_{i-1}\}} + \chi^{\{v_{i-1}, v_i\}}\enspace.
	\end{align*}
	Here, we take the sums cyclically around $C$ and use alternating signs (see \cref{fig:facet-cut-edge} for a visualization).
	All terms on the left-hand side of this equality are in $F$.
	Additionally, we already showed that $\chi^{\{v_j, v_{j+1}\}}\in S$ for all $j\in [2k+1]$.
	Hence, also $\chi^{\{v_i,s\}}\in S$.
	\begin{figure}[ht]
		\begin{subfigure}{0.48\textwidth}
			\begin{center}
			\begin{tikzpicture}[scale=0.5, rotate=-90]

				\pgfdeclarelayer{background}
				\pgfsetlayers{background,main}

				\begin{scope}[every
					node/.style={node_style}]
					\node (v) at (0,6) {};
				\end{scope}

				\node[above=2pt] at (v) {$s$};

				\begin{pgfonlayer}{background}
					\node[regular polygon, regular polygon sides=7, minimum size=3cm, rotate=-90]
					(a) {};
				\end{pgfonlayer}

				\foreach \i in {1,2,...,7} {
					\node[node_style] (\i) at (a.corner \i) {};
				}
				\node (8) at (a.corner 1) {};

				\begin{pgfonlayer}{background}
					\foreach \i in {1,3,...,7} {

						\pgfmathtruncatemacro{\j}{\i+1 }

						\draw[line width=1pt, darkgreen, fill=darkgreen, fill
						opacity=0.5] (v.center) -- (\i.center) -- (\j.center) -- (v.center);
					}
				\end{pgfonlayer}

		\node[regular polygon, regular polygon sides=7, minimum size=3.8cm, rotate=-90] (labels) {};
		\foreach \i in {2,3,...,7} {
			\node (\i) at (labels.corner \i) {$v_\i$};
		}
		\node[xshift=-0.9cm] (1) at (labels.corner 1) {$v_1$};
			\end{tikzpicture}
 		\end{center}
		\end{subfigure}
		\hfill
		\begin{subfigure}{0.48\textwidth}
			\begin{center}
			\begin{tikzpicture}[scale=0.5, rotate=-90]

				\pgfdeclarelayer{background}
				\pgfsetlayers{background,main}
				\begin{scope}[every
					node/.style={node_style}]
					\node (v) at (0,6) {};
				\end{scope}

				\node[above=2pt] at (v) {$s$};

				\begin{pgfonlayer}{background}
					\node[regular polygon, regular polygon sides=7, minimum	size=3cm, rotate=-90] (a) {};
				\end{pgfonlayer}

				\foreach \i in {1,2,...,7} {
					\node[node_style] (\i) at (a.corner \i) {};
				}
				\node (8) at (a.corner 1) {};

				\begin{pgfonlayer}{background}
					\foreach \i in {2,4,6} {

						\pgfmathtruncatemacro{\j}{\i+1 }

						\draw[line width=1pt, darkred, fill=red, fill=darkred,
						fill opacity=0.5] (v.center) --  (\i.center) --
						(\j.center) -- (v.center);
					}
				\end{pgfonlayer}

		\node[regular polygon, regular polygon sides=7, minimum size=3.8cm, rotate=-90] (labels) {};
		\foreach \i in {2,3,...,7} {
			\node (\i) at (labels.corner \i) {$v_\i$};
		}
		\node[xshift=-0.9cm] (1) at (labels.corner 1) {$v_1$};
			\end{tikzpicture}
 		\end{center}
		\end{subfigure}
		\caption{To show that $\chi^{\{v_i, s\}}\in S$, we use an alternating sum of incidence vectors of triangles.
			Adding the incidence vectors of all green triangles on the left and subtracting those of the red ones on the right, we end up with a linear combination containing twice the edge $\{v_1,s\}$, and additionally the edges of $C$ with alternating coefficients $1$ and $-1$.}
		\label{fig:facet-cut-edge}
	\end{figure}
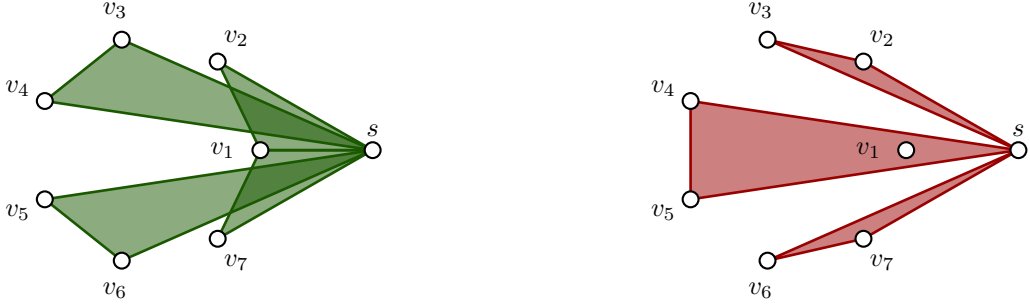

	Altogether, we thus showed $\chi^f\in S$ for every $f\in E$, hence we conclude $\dim(S) = |E|$, finishing the proof.
\end{proof}

We remark that the $\mathcal{C}$-induced facets are also facets of $P_{\textup{odd}}(K_n)$, as the coefficient vector is strictly positive.
(To see this, one could also observe that the arguments in the above proof only exploit characteristic vectors of odd cycles, and not of general points in the dominant.)

\subsection{Extending to facets of the odd matching polytope}

In the following, we prove \cref{lem:expressing-facet-transfer,lem:n-1expressing}, showing that we can transform the $\mathcal{C}$-induced facets of $\domP(K_n)$ into facets of the odd-red bipartite perfect matching polytope $P_{(\widehat{K}_n, \widehat{R})}$.
We begin with providing the proof of \cref{lem:expressing-facet-transfer},~i.e., we show that transforming an $(n-1)$-expressible facet-defining constraint of $\domP(G)$ yields a facet-defining constraint of $P_{(\widehat{G}, \widehat{R})}$.

\begin{proof}[Proof of \cref{lem:expressing-facet-transfer}]
	Recall that we consider a graph $G = (V,E)$ and the corresponding complete bipartite red-blue colored graph $\widehat{G} = (\widehat{V}, \widehat{E})$ with red edges $\widehat{R}$.
	Let $\kappa$ be an $(n-1)$-expressible constraint of $\domP(G)$ that defines the facet $F$, and denote by $\widehat{\kappa}$ the canonical transformation of $\kappa$.
	To show that
	\begin{align*}
		\widehat{F} \coloneqq \{x\in P_{(\widehat{G}, \widehat{R})}\colon \widehat{\kappa}\text{ is tight for $x$}\}
	\end{align*}
	is a facet of $P_{(\widehat{G}, \widehat{R})}$, we argue about the dimension of $\widehat{F}$.
	To start with, note that because $\kappa$ is facet-defining for $\domP(G)$, there is an odd cycle in $G$ whose characteristic vector is not tight for $\kappa$.
	In particular, transforming such a cycle into an odd-red perfect matching of $\widehat{G}$ yields a point in $P_{(\widehat{G}, \widehat{R})}$ that is not contained in $\widehat{F}$.
	Consequently, $\widehat{F}$ is a proper face of $P_{(\widehat{G}, \widehat{R})}$, and we have $\dim(\widehat{F}) < \dim(P_{(\widehat{G}, \widehat{R})})$.
	Also, because $0\notin\widehat{F}$, we have $\dim(\widehat{F}) = \dim(\linspan(\widehat{F}))-1$.
	Finally, $P_{(\widehat{G}, \widehat{R})}$ is contained in the perfect matching polytope of $\widehat{G}$, hence its dimension is at most $|\widehat{E}| - |\widehat{V}| + 1$.
	Together, these observations yield
	\[
		\dim(\linspan(\widehat{F})) - 1 = \dim(\widehat{F}) < \dim(P_{(\widehat{G}, \widehat{R})}) \leq |\widehat{E}| - |\widehat{V}| + 1\enspace.
	\]
	To conclude that $\widehat{F}$ is a facet of $P_{(\widehat{G}, \widehat{R})}$, it thus suffices to show that
	\[
		\dim(\linspan(\widehat{F})) \geq |\widehat{E}| - |\widehat{V}| + 1 = |\widehat{E}| - (2|V| - 1)\enspace.
	\]
	In the following, we achieve this by constructing a set $H\subseteq \mathbb{R}^{\widehat{E}}$ with $|H| = 2|V| - 1$ and $\linspan(\widehat{F} \cup H) = \mathbb{R}^{\widehat{E}}$.
	To build such a set $H$, we use the fact that $\kappa$ is $(n-1)$-expressing.
	Let $Q \subseteq E$ be a set of edges such that all of $E$ is $Q$-expressible and $|Q| = n - 1$.
	For every edge $e = \{v,w\} \in E$, let $\overrightarrow{e} \in \widehat{E}$ be an arbitrary but fixed one of the two edges $\{v^+, w^-\}$ and $\{v^-, w^+\}$ in $\widehat{G}$ that correspond to $e$.
	We define
	\[
		H \coloneqq \{\chi^{\overrightarrow{e}}\colon e \in Q\} \cup \{\chi^{\{{v}^+, v^-\}}\colon v\in V\}
	\]
	and claim that this set $H$ has the desired property.
	To show that $\linspan(\widehat{F} \cup H) = \mathbb{R}^{\widehat{E}}$, we show that $\chi^{\widehat{e}}\in \linspan(\widehat{F} \cup H)$ for every $\widehat{e}\in \widehat{E}$.
	First, this is clear if $\widehat{e} = \{v^+, v^-\}$ for some $v\in V$, as then $\chi^{\widehat{e}}\in H$ by definition.
	To deal with the remaining cases, we will show that for every $Q$-expressible edge $e=\{v,w\}\in E$, both edges $\{v^+, w^-\}$ and $\{v^-, w^+\}$ corresponding to $e$ have their characteristic vectors in $\linspan(\widehat{F} \cup H)$.
	Thus, let $e=\{v,w\}\in E$ be a $Q$-expressible edge.
	There are two cases to consider: either $e\in Q$, or there is an odd cycle $C$ in $G$ that---apart from~$e$---contains only $Q$-expressible edges such that $\chi^C$ is tight for $\kappa$.
	The first case is covered by the following claim.

	\begin{claim}\label{claim:transfer-expressibility-base}
		For $e=\{v,w\}\in Q$, both $\chi^{\{v^+, w^-\}}$ and $\chi^{\{v^-, w^+\}}$ are in $\linspan(\widehat{F} \cup H)$.
	\end{claim}

	\begin{proof}[Proof of \cref{claim:transfer-expressibility-base}]
		One of the two characteristic vectors, say $\chi^{\{v^+, w^-\}}$ without loss of generality, lies in $H$ by definition of $H$.
		To show that the other lies in $\linspan(\widehat{F} \cup H)$ as well, it is enough to show that $\chi^{\{v^+, w^-\}} + \chi^{\{v^-, w^+\}} \in \linspan(\widehat{F})$.

		Because $\kappa$ is facet-defining for $\domP(G)$ and $\domP(G)$ is full-dimensional, we have $\linspan(\widehat{F}) = \mathbb{R}^{\widehat{E}}$, and we can therefore write $\chi^{e}$ as a linear combination
		\begin{align}\label{eq:linear-combination-odd-cycles}
			\chi^e = \sum_{C\in\mathcal{C}} \alpha_C \chi^C\enspace,
		\end{align}
		where $\mathcal{C}$ is a set of odd cycles that are tight for $\kappa$.
		For every $C\in\mathcal{C}$, define a point $y^C\in \mathbb{R}^{\widehat{E}}$ by
		\[
		y^C({\widehat{f}}) \coloneqq
		\begin{cases}
			\frac{1}{2} & \text{if } \widehat{f} = \{x^+, y^-\} \text{ or } \widehat{f} = \{x^-, y^+\} \text{ for some } \{x,y\}\in C\\
			1 & \text{if } \widehat{f} = \{{u}^+, {u}^-\} \text{ for some } u\in V\setminus V(C)\\
			0 & \text{else}
		\end{cases}
		\]
		for all $\widehat{f}\in \widehat{E}$.
		From~\eqref{eq:linear-combination-odd-cycles}, we obtain
		\begin{align}\label{eq:linear-combination-odd-cycles-transfer}
			\frac12\left(\chi^{\{v^+, w^-\}} + \chi^{\{v^-, w^+\}}\right) = \sum_{C\in\mathcal{C}} \alpha_C y^C\enspace.
		\end{align}
		Note that $y^C$ is tight for $\widehat{\kappa}$ by definition of the canonical transformation, and thus, in particular, $y^C\in \widehat{F}$.
		Together with~\eqref{eq:linear-combination-odd-cycles-transfer}, this implies $\chi^{\{v^+, w^-\}} + \chi^{\{v^-, w^+\}} \in \linspan(\widehat{F})$, as desired.
	\end{proof}

	Next, we deal with the case that $e\notin Q$.
	We first establish the following claim.

	\begin{claim}\label{claim:transfer-expressibility}
		Let $C$ be a cycle in $G$ that is tight for $\kappa$, and let $e=\{v,w\}\in C$.
		If, for every edge $f=\{x,y\}\in C\setminus \{e\}$, both $\chi^{\{x^+, y^-\}}$ and $\chi^{\{x^-, y^+\}}$ are in $\linspan(\widehat{F} \cup H)$, then also both $\chi^{\{v^+, w^-\}}$ and $\chi^{\{v^-, w^+\}}$ are in $\linspan(\widehat{F} \cup H)$.
	\end{claim}

	\begin{proof}[Proof of \cref{claim:transfer-expressibility}]
		We show that $\chi^{\{v^+, w^-\}}\in \linspan(\widehat{F} \cup H)$; the rest follows by symmetry.
		Let $\overrightarrow{C} = (w_1, \dots, w_t)$ be the orientation of $C$ that contains the arc $(v,w)$.
			Define a point $y^{\overrightarrow{C}}\in \mathbb{R}^{\widehat{E}}$ by
			\[
			y^{\overrightarrow{C}}({\widehat{f}}) \coloneqq
			\begin{cases}
				1 & \text{if } \widehat{f} = \{w_i^+, w_{i+1}^-\} \text{ for } i\in [t],\\
				1 & \text{if } \widehat{f} = \{u^+, u^-\} \text{ for } u\in V\setminus V(C),\\
				0 & \text{else,}
			\end{cases}
			\]
			for all $\widehat{f}\in \widehat{E}$, where we set $w_{t+1} \coloneqq w_1$.
			Because $C$ is tight for $\kappa$, $y^{\overrightarrow{C}}$ is tight for $\widehat{\kappa}$, and hence $y^{\overrightarrow{C}}\in \widehat{F}$.
			Observe that
			\[
				y^{\overrightarrow{C}} = \sum_{i=1}^t \chi^{\{w_i^+, w_{i+1}^-\}} + \sum_{u\in V\setminus V(C)} \chi^{\{{u}^+, {u}^-\}}\enspace.
			\]
			Here, all terms but $\chi^{\{v^+, w^-\}}$ are in $\linspan(\widehat{F} \cup H)$, hence also $\chi^{\{v^+, w^-\}}\in \linspan(\widehat{F} \cup H)$.
		\end{proof}

		Using \cref{claim:transfer-expressibility}, we see that expressibility of edges $\{v, w\}$ in $G$ directly translates to the characteristic vectors $\chi^{\{v^+, w^-\}}$ and $\chi^{\{v^-, w^+\}}$ lying in $\linspan(\widehat{F} \cup H)$.
		Because every edge of $G$ is $Q$-expressible, we thus conclude that $\chi^{\widehat{e}}\in \linspan(\widehat{F} \cup H)$ for every $\widehat{e}\in \widehat{E}$, finishing the proof.
\end{proof}

To finish the transformation, we need to prove \cref{lem:n-1expressing}.

\begin{proof}[Proof of \cref{lem:n-1expressing}]
	Consider a cycle $C$ of length $n-2$ in $K_n$, and let $\kappa$ be the corresponding $C$-induced constraint.
	Let $s$ and $t$ be the two vertices in $K_n$ that are not contained in $C$.
	We claim that every edge in $E$ is $Q$-expressible for $Q = \delta(s)$.
	Before we show this, note that $|Q| = n - 1$, hence the claim indeed implies that $\kappa$ is $(n-1)$-expressing.
	To prove the claim, let $e\in E$ be an arbitrary edge.
	We distinguish the following cases:
	\begin{itemize}
		\item If $e = \{s,v\}$ for some $v\in V(C) \cup \{t\}$, then $e\in Q$ and thus $e$ is trivially $Q$-expressible.
		\item If $e = \{t,v\}$ for some $v\in V(C)$, then consider the triangle $T_v = (s, t, v)$.
			Note that $\chi^{T_v}$ is tight for $\kappa$.
			Additionally, $\{s,v\},\{s,t\}\in Q$, which implies in particular that these edges are $Q$-expressible.
			Hence, by definition of $Q$-expressibility, $e$ is $Q$-expressible as well.
		\item If $e = \{v_i, v_{i+1}\}$ for $i\in [n-3]$, then consider the triangle $T^s_i = (s, v_i, v_{i+1})$.
			Again, $\chi^{T^s_i}$ is tight for $\kappa$.
			Additionally, $\{s,v_i\},\{s,v_{i+1}\}\in Q$, and we can again conclude that $e$ is $Q$-expressible.
		\item If $e = \{v_i, v_j\}$ for non-adjacent vertices $v_i, v_j$ of $C$, then let $C_{i,j}$ be the cycle consisting of $e$ and the even length path between $v_i$ and $v_j$ on $C$.
			As before, $\chi^{C_{i,j}}$ is tight for $\kappa$, and all edges of $C$ except for $e$ are $Q$-expressible by the previous case, so we conclude that $e$ is $Q$-expressible, as well.
	\end{itemize}
	Consequently, every $e\in E$ is $Q$-expressible, as desired.
\end{proof}

\subsection[Asserting complexity of canonical transformations of \texorpdfstring{$\mathcal{C}$}{C}-induced facets]{Asserting complexity of canonical transformations of \texorpdfstring{\boldmath$\mathcal{C}$\unboldmath}{C}-induced facets}

In this section, we conclude the proof of \cref{thm:diverse_and_large_coeff} by showing that canonical transformations of $\mathcal{C}$-induced facets do not admit significantly simpler representations than the original $\mathcal{C}$-induced facets themselves.
We again start from the complete graph $G=K_n$ on $n=2k+3$ vertices, and let $C$ be a cycle of length $n-2$ in $G$; $s$ and $t$ are the two vertices in $V$ not contained in $C$.
Let $a\in\mathbb{Z}^{\widehat E}$ and $b\in\mathbb{Z}$ be such that $a^\top y \leq b$ is the canonical transformation of the $C$-induced constraint, where, as before, $\widehat{G}=(\widehat{V}, \widehat{E})$ is the complete bipartite graph on vertex set $\widehat{V}=V^+ \cup V^-$ with red edges $\widehat{R}$ resulting from the transformation of~$G$.
We recall that the representation of the constraint is not unique, not even up to scaling, since the polytope $P_{(\widehat{G}, \widehat{R})}$ is not full-dimensional.
Concretely, one may add any linear combination of the degree (equality) constraints to $a^\top y \leq b$ without changing the face defined by the constraint.%
\footnote{In fact, one can observe that up to scaling, any equivalent representation of a facet-defining constraint $a^\top y \leq b$ can be obtained by adding a linear combination of degree constraints, as the coefficient vectors of degree constraints span the linear space orthogonal to the affine hull of $P_{(\widehat{G}, \widehat{R})}$.}
Our goal is to show that no such transformation can yield a significantly simpler representation of the facet, i.e., a representation with significantly smaller or fewer different coefficients.

The aforementioned constraint transformations are particularly convenient to deal with when viewing the coefficient vector $a$ as a matrix in $\mathbb{Z}^{V^+\times V^-}$, where for $u\in V^+$ and $v\in V^-$, the matrix entry $a(u, v)$ is the coefficient $a_{\{u, v\}}$ of the variable associated with the edge $\{u, v\}\in \widehat E$.
To be explicit, we have
\begin{equation*}
	a(u^+, v^-) 
	= \begin{cases}
		0 & \text{if } u=v\\
		\bigl| 2k+1 - 2|i-j|\bigr| & \text{if } u=v_i, v=v_j, i\neq j\\
		k & \text{if } |\{u,v\}\cap \{s,t\} |= 1 \\
		1 & \text{if } \{u,v\} = \{s,t\}
	\end{cases}\enspace.
\end{equation*}
A presentation of $a$ in matrix form can be found in \cref{fig:matrix-example-2}.

\begin{figure}[ht]
	\centering
	\setlength{\tabcolsep}{0.05cm}
	\begin{tabular}{c|@{\hspace{0.1cm}}ccccccc@{\hspace{0.2cm}}cc}
		& $v_1^-$ & $v_3^-$ & $v_5^-$ & $v_7^-$ & $v_2^-$ & $v_4^-$ & $v_6^-$ & $s^-$ & $t^-$ \\[0.05cm]
		\hline
		\rule{0pt}{1.2em}
		$v_1^+$ & 0 & 5 & 3 & 1 & 1 & 3 & 5 & 3 & 3 \\
		$v_3^+$ & 5 & 0 & 5 & 3 & 1 & 1 & 3 & 3 & 3 \\
		$v_5^+$ & 3 & 5 & 0 & 5 & 3 & 1 & 1 & 3 & 3 \\
		$v_7^+$ & 1 & 3 & 5 & 0 & 5 & 3 & 1 & 3 & 3 \\
		$v_2^+$ & 1 & 1 & 3 & 5 & 0 & 5 & 3 & 3 & 3 \\
		$v_4^+$ & 3 & 1 & 1 & 3 & 5 & 0 & 5 & 3 & 3 \\
		$v_6^+$ & 5 & 3 & 1 & 1 & 3 & 5 & 0 & 3 & 3 \\[0.2cm]
		$s^+$   & 3 & 3 & 3 & 3 & 3 & 3 & 3 & 0 & 1 \\
		$t^+$   & 3 & 3 & 3 & 3 & 3 & 3 & 3 & 1 & 0 \\
	\end{tabular}
	\caption{Example of a coefficient vector $a$ in matrix form for $n=9$.}
	\label{fig:matrix-example-2}
\end{figure}

Generally, a transformation of the constraint $a^\top y \leq b$ through scaling and addition of linear combinations of the degree constraints results in a new constraint of the form
\begin{equation*}
	\mu \cdot a^\top y + \sum_{v\in \widehat V} \lambda_v \cdot y(\delta(v)) \geq \mu \cdot b + \sum_{v\in \widehat V} \lambda_v
\end{equation*}
with $\mu \in \mathbb{R}\setminus \{0\}$ and $\lambda \in \mathbb{R}^{\widehat V}$, and a corresponding coefficient vector $a^{\mu, \lambda}$ defined by
\begin{equation*}
	a^{\mu,\lambda}(u^+, v^-) = \mu \cdot a(u^+, v^-) + \lambda_{u^+} + \lambda_{v^-} \quad \text{for all } u,v\in V\enspace.
\end{equation*}
Thus, in matrix form, a constraint transformation of this type corresponds to scaling the matrix by a factor $\mu$, adding a constant $\lambda_{u^+}$ to all entries in the row corresponding to $u^+$ for every $u\in V$, and adding a constant $\lambda_{v^-}$ to all entries in the column corresponding to $v^-$ for every $v\in V$.
We may further restrict our attention to transformations that result in integral matrix entries.
For these, we prove the following complexity result.

\begin{lemma}\label{lem:complexity-C-induced}
	Let $a\in \mathbb{Z}^{\widehat{E}}$ and $b\in \mathbb{Z}$ such that $a^\top y \leq b$ is the canonical transformation of a $\mathcal{C}$-induced constraint of $\domP(K_n)$.
	Let $\mu \in \mathbb{R}\setminus \{0\}$ and $\lambda\in \mathbb{R}^{\widehat{V}}$ be such that $a^{\mu,\lambda}$ is integral.
	Then the following holds.
	\begin{enumerate}
		\item\label{item:small_coeffs} The largest absolute value of an entry in $a^{\mu,\lambda}$ is at least $\frac{n-4}{2}$.
		\item\label{item:diverse_coeffs} The number of distinct entries in $a^{\mu,\lambda}$ is at least $\sqrt{\frac{n-1}{2}}$.
	\end{enumerate}
\end{lemma}

\begin{proof}
	Fix a cycle $C = (v_1, v_2, \dots, v_{2k+1})$ of length $n-2$ in $K_n$, and let $s$ and $t$ be the two vertices in $V$ not contained in $C$.
	We consider the canonical transformation $a^\top y \leq b$ of the corresponding $C$-induced constraint, and show the results for arbitrary fixed $\mu\in\mathbb{R}\setminus\{0\}$ and $\lambda\in \mathbb{R}^{\widehat{V}}$ such that $a^{\mu,\lambda}$ is integral.

	Let us first argue why not all entries in $a^{\mu,\lambda}$ can be small in absolute value, i.e., prove \cref{item:small_coeffs}. 
	To this end, we repeatedly use that for any sequence of vertices $(u_1, u_2, \dots, u_t)$ in $G$, defining $u_{t+1}\coloneqq u_1$, we have
	\begin{equation}\label{eq:alternating_sum}
		\begin{aligned}
			\sum_{i=1}^t a(u_i^+, u_{i+1}^-) &= \sum_{i=1}^t a(u_i^+, u_{i+1}^-) - \sum_{i=1}^t a(u_i^+, u_i^-) \\
			&= \frac{1}{\mu} \left(\sum_{i=1}^t a^{\mu,\lambda}(u_i^+, u_{i+1}^-) - \sum_{i=1}^t a^{\mu,\lambda}(u_i^+, u_i^-)\right)\enspace.
		\end{aligned}
	\end{equation}
	Here, the first equality holds because $a(u_i^+, u_i^-) = 0$ for all $i\in [t]$, and the second equality follows from the definition of $a^{\mu,\lambda}$ and the fact that the contributions of the row and column additions cancel out in the difference because the right-hand side sum has one term with positive sign and one term with negative sign per involved row and column of $a^{\mu,\lambda}$.
	In particular, from \eqref{eq:alternating_sum}, we conclude the following:
	\begin{enumerate}
		\item\label{item:denominator} As $a$ and $a^{\mu,\lambda}$ are integral, $\mu$ is a rational number and its denominator (in fully reduced form) divides the left-hand side $\sum_{i=1}^t a(u_i^+, u_{i+1}^-)$ of \eqref{eq:alternating_sum}.
		\item\label{item:absvalue} Because the maximum of several nonnegative numbers is at least their average, and by the triangle inequality, we have 
			\begin{align*}
				\max_{i \in [t]} \max &\{|a^{\mu,\lambda}(u_i^+, u_{i+1}^-)|, |a^{\mu,\lambda}(u_i^+, u_i^-)|\} \\
				&\geq \frac{\sum_{i=1}^t \left|a^{\mu,\lambda}(u_i^+, u_{i+1}^-)\right| + \sum_{i=1}^t \left|a^{\mu,\lambda}(u_i^+, u_i^-)\right|}{2t} \\
				&\geq \frac{\big|\sum_{i=1}^t a^{\mu,\lambda}(u_i^+, u_{i+1}^-) - \sum_{i=1}^t a^{\mu,\lambda}(u_i^+, u_i^-)\big|}{2t}
				\stackrel{\text{\eqref{eq:alternating_sum}}}{=} \frac{|\mu|}{2t} \cdot \left|\sum_{i=1}^t a(u_i^+, u_{i+1}^-)\right| .
			\end{align*}
	\end{enumerate}
	Concretely, this lets us draw the following conclusions.
	\begin{itemize}
		\item \emph{Sequence $(s,t)$:} We have $a(s^+, t^-) = a(t^+, s^-) = 1$, and thus $a(s^+, t^-) + a(t^+, s^-) = 2$, hence the denominator of $\mu$ divides $2$ by \cref{item:denominator}.
		\item \emph{Sequence $(v_1, v_2, v_3)$:} We have $a(v_1^+, v_2^-) = 1$, $a(v_2^+, v_3^-) = 1$, and $a(v_3^+, v_1^-) = 2k-1$, hence $a(v_1^+, v_2^-) + a(v_2^+, v_3^-) + a(v_3^+, v_1^-) = 2k+1$. 
			Consequently, the denominator of $\mu$ divides $2k+1$ by \cref{item:denominator}.
			Together with the previous point, we conclude that the denominator of $\mu$ must be $\pm 1$, so $\mu$ must be an integer.
			In particular, we have $|\mu| \geq 1$.
		\item \emph{Sequence $(v_1, v_3)$:} We have $a(v_1^+, v_3^-) = a(v_3^+, v_1^-) = 2k-1$, hence $a(v_1^+, v_3^-) + a(v_3^+, v_1^-) = 4k-2$.
			By \cref{item:absvalue}, we get
			\[
				\max \{|a^{\mu,\lambda}(v_1^+, v_3^-)|, |a^{\mu,\lambda}(v_3^+, v_1^-)|\} \geq \frac{|\mu|}{4} (4k-2) \geq \frac{n-4}{2}\enspace.
			\]
	\end{itemize}
	
	Next, we show that $a^{\mu,\lambda}$ contains many different entries.
	To this end, consider the two rows of $a$ indexed by $v_1^+$ and $s^+$.
	Note that in the columns of $v_1^-$, $v_2^-$, \dots, $v_{k+1}^-$, the row of $v_1^+$ contains the values $\{0, 1, 3, 5, \dots, 2k-1\}$ once each, while the row of $s^+$ contains only the value $k$.
	As the $s^+$-row of $a^{\mu,\lambda}$ is obtained from the $s^+$-row of $a$ by scaling and adding $\lambda_{v_i^-}$ to the entry corresponding to column $v_i^-$, we have
	\begin{equation*}
		|\{a^{\mu,\lambda}(s^+, v_i^-)\colon i \in [k+1]\}| \geq |\{\lambda_{v_i^-}\colon i\in [k+1]\}| \eqqcolon r \enspace.
	\end{equation*}
	On the other hand, by the pigeonhole principle, there is a value $L$ that appears at least $\lceil\sfrac{k+1}{r}\rceil$ times among the values $(\lambda_{v_1^-}, \ldots, \lambda_{v_{k+1}^-})$.
	Let $I\subseteq [k+1]$ be the set of indices $i$ such that $\lambda_{v_i^-} = L$.
	Then, the values of $a^{\mu,\lambda}(s^+, v_i^-)$ for $i\in I$ are all different, as $a^{\mu,\lambda}(s^+, v_i^-) = a(s^+, v_i^-) - \lambda_{s^+} - L$.
	Thus, the number of different values in $a^{\mu,\lambda}$ is at least $|I| \geq \lceil\sfrac{k+1}{r}\rceil$.
	Combining both bounds, we obtain that $a^{\mu,\lambda}$ contains at least $\max\{r, \lceil\sfrac{k+1}{r}\rceil\}$ different values.
	Regardless of the value of $r$, this is at least $\sqrt{k+1} = \sqrt{\frac{n-1}{2}}$, finishing the proof.
\end{proof}

Now we are ready to prove \cref{thm:diverse_and_large_coeff}.
\begin{proof}[Proof of \cref{thm:diverse_and_large_coeff}]
	Let $G=(V,E)$ be the complete graph on $n=2k+3$ vertices and let $C$ be a cycle of length $n-2$ in $G$.
	By \cref{lem:c-induced-facet}, the canonical transformation of the $C$-induced constraint is facet-defining for $P_{(\widehat{G}, \widehat{R})}$.
	Let the canonical transformation be the constraint $a^\top y \leq b$ defined by $a\in \mathbb{Z}^{\widehat{E}}$ and $b\in \mathbb{Z}$.
	Any other equivalent representation of the constraint has a coefficient vector of the form $a^{\mu,\lambda}$ for some $\mu\in \mathbb{R}\setminus \{0\}$ and $\lambda\in \mathbb{R}^{\widehat{V}}$. 
	By \cref{lem:complexity-C-induced}, such integral coefficients satisfy that
	\begin{enumerate}
		\item the largest absolute value of an entry in $a^{\mu,\lambda}$ is at least $\frac{n-4}{2}$, and
		\item the number of distinct entries in $a^{\mu,\lambda}$ is at least $\sqrt{\frac{n-1}{2}}$.
	\end{enumerate}
	To finish the proof, it remains to show that canonical transformations of $\mathcal{C}$-induced facets indeed result in exponentially many different facets of $P_{(\widehat{G}, \widehat{R})}$.
	To see this, note that for two distinct cycles $C_1, C_2 \in \mathcal{C}$, we have that $\chi^{C_1}$ is tight for the $C_1$-induced constraint but not for the $C_2$-induced constraint, hence the constraints define different facets.
	This carries over to the canonical transformations of $\mathcal{C}$-induced constraints.
	Indeed, let $x^{C_1} \in P_{(\widehat{G}, \widehat{R})}$ be the point defined by
	\[
		x^{C_1}({\widehat{e}}) \coloneqq
		\begin{cases}
			\frac{1}{2} & \text{if } \widehat{e} = \{u^+, v^-\} \text{ for } \{u,v\}\in C_1,\\
			1 & \text{if } \widehat{e} = \{{w}^+, {w}^-\} \text{ for } w\in V\setminus V(C_1),\\
			0 & \text{else,}
		\end{cases}
	\]
	for all $\widehat{e}\in \widehat{E}$.
	Then $x^{C_1}\in P_{(\widehat{G}, \widehat{R})}$ and $x^{C_1}$ is tight for the canonical transformation of the $C_1$-induced constraint, but not for the canonical transformation of the $C_2$-induced constraint.
	Consequently, the corresponding facets of $P_{(\widehat{G}, \widehat{R})}$ are different, and we constructed an injective mapping from $\mathcal{C}$ to the set of facets of $P_{(\widehat{G}, \widehat{R})}$ that are canonical transformations of $\mathcal{C}$-induced facets.
	Since there are exponentially many different cycles in $\mathcal{C}$, we conclude that we indeed obtain exponentially many facets with the claimed complexity.
\end{proof}

In the remainder of this section, we complement the result of \cref{thm:diverse_and_large_coeff} by showing that the complexity of canonical transformations of $\mathcal{C}$-induced facets in terms of the number of different coefficients that are required can indeed be reduced to a sublinear bound, concretely to $O(n^{\sfrac{2}{3}})$ different coefficients.
This is formalized below.

\begin{theorem}\label{thm:upper-bound}
    Let $F$ be a facet of $P_{(\widehat G, \widehat R)}$ defined by the canonical transformation of a $\mathcal{C}$-induced constraint of $\domP(G)$ for $G=K_n$ with odd $n\in\mathbb{Z}_{\geq 3}$.
    Then, there exists a coefficient vector $\widehat a\in \mathbb{Z}^{\widehat E}$ and $\widehat b\in \mathbb{Z}$ such that $F = \{y\in P_{(\widehat G, \widehat R)}\colon \widehat a^\top y = \widehat b\}$ and
	\[
		\left|\left\{\widehat a_{\widehat e}\colon \widehat e \in \widehat E\right\}\right| = O(n^{\sfrac{2}{3}})\enspace.
	\]
\end{theorem}

\begin{proof}
	Let $G=(V,E)$ be the complete graph on $n=2k+3$ vertices, and let $C=(v_1, \dots, v_{2k+1})$ be a cycle of length $n-2$ in $G$.
	Let $\{s,t\} = V\setminus V[C]$.
	Recall that we can write the canonical transformation of the $C$-induced constraint as $a^\top y \leq b$ with $a\in \mathbb{Z}^{\widehat E}$ given by
	\begin{equation*}
		a(u^+, v^-) 
		= \begin{cases}
			0 & \text{if } u=v\\
			\bigl| 2k+1 - 2|i-j|\bigr| & \text{if } u=v_i, v=v_j, i\neq j\\
			k & \text{if } |\{u,v\}\cap \{s,t\} |= 1 \\
			1 & \text{if } \{u,v\} = \{s,t\}
		\end{cases}
	\end{equation*}
	for all $u,v\in V$.
	We now explicitly construct a vector $\lambda\in \mathbb{R}^{\widehat{V}}$ for which $a^{1,\lambda}$ has $O(n^{\frac{2}{3}})$ different entries, i.e., we use $\mu = 1$ and only add linear combinations of degree constraints to the original constraint.
	Our vector $\lambda$ takes the form
	\[
		\lambda_w = \begin{cases}
			0 & \text{if } w \in \{s^+, s^-, t^+, t^-\}\\
			2f_i & \text{if } w = v_i^+ \text{ for some } i\in [2k+1]\\
			-2f_i & \text{if } w = v_i^- \text{ for some } i\in [2k+1]
		\end{cases}
	\]
	for some $f_i \in \mathbb{Z}$ to be defined later for all $i\in [2k+1]$.
	Taking into account that $a^{1, \lambda}(u^+, v^-) = a(u^+, v^-) + \lambda_{u^+} + \lambda_{v^-}$ for all $u,v\in V$, we have that
	\begin{multline*}
		\{a^{1,\lambda}(u^+, v^-)\colon u,v\in V\} = \{0, 1\} \cup \{|2k+1 - 2|i-j|| + 2f_i - 2f_j \colon i,j\in [2k+1], i\neq j\} \\
		\cup \{k + 2f_i\colon i\in [2k+1]\} \cup \{k - 2f_i\colon i\in [2k+1]\}\enspace.
	\end{multline*}
	Here, the first set stems from entries $a^{1,\lambda}(u^+, u^-)$ with $u\in V$, the second set from entries $a^{1,\lambda}(v_i^+, v_j^-)$ for distinct $i, j\in[2k+1]$, and the last two sets from entries $a^{1,\lambda}(u^+, v^-)$ with $|\{u,v\} \cap \{s,t\}| = 1$.
	Note that
	\[
		|2k+1 - 2|i-j|| + 2f_i - 2f_j = \pm (2k+1) \pm 2(i-j) + 2f_i - 2f_j\enspace,
	\]
	where the combination of signs depends on whether $j < i$ and whether $|i-j|$ is larger than $k$.
	The shift by $\pm(2k+1)$ can at most double the number of different values that this expression takes compared to the number of different values taken by $\pm 2(i-j) + 2f_i - 2f_j$.
	Going further, the latter is either $2(f_i+i) - 2(f_j+j)$ or $2(f_i - i) - 2(f_j - j)$, depending on the combination of signs.
	We thus conclude an upper bound on the number of different entries in $a^{1,\lambda}$ of the form
	\begin{multline}\label{eq:diffcount}
		\left|\{a^{1,\lambda}(u^+, v^-)\colon u,v\in V\}\right|
		\leq 2 + 2 \cdot |\{f_i: i\in[2k+1]\}| + |\{(f_i + i) - (f_j + j)\colon i,j\in[2k+1], i\neq j\}| \\
		+ |\{(f_i - i) - (f_j - j)\colon i,j\in[2k+1], i\neq j\}|\enspace.
	\end{multline}
	Our concrete choice of the values $f_i$ is determined by the following lemma, which we will apply with $m_1 = \lceil n^{\sfrac{1}{3}}\rceil$ and $m_2 = \lceil n^{\sfrac{2}{3}}\rceil$.
	Here, for $a,b\in \mathbb{Z}$, we denote by $(a \bmod b) \in \{0, \dots, b-1\}$ the remainder of $a$ when divided by $b$.

	\begin{lemma}\label{lem:shift-properties}
		For all $m_1, m_2 \in \mathbb{Z}$, there exist an integral sequence $(f_j)_{j\geq 0}$ such that for all $j\in \mathbb{Z}_{\geq 0}$, we have
		\begin{enumerate}[label=\normalfont(\roman*)]
			\item\label{item:bounded} $0 \leq f_j \leq m_1 + m_2$,
			\item\label{item:divisible} $f_j + j$ is divisible by $m_1$, and
			\item\label{item:residue} $((f_j - j) \bmod m_2) \in \{0, 1, \dots, m_1\}$.
		\end{enumerate}
	\end{lemma}

	Before proving the lemma, we show how these numbers $f_j$ can be used for our application with $m_1 = \lceil n^{\sfrac{1}{3}}\rceil$ and $m_2 = \lceil n^{\sfrac{2}{3}}\rceil$.
	First, \cref{item:bounded} directly implies that $|\{f_i\colon i\in [2k+1]\}| \leq m_2 + m_1 + 1 \in O(n^{\sfrac{2}{3}})$.
	Next, \cref{item:divisible} implies that $(f_i + i) - (f_j + j)$ is divisible by $m_1$ for all $i,j\in [2k+1]$, while being of absolute value at most $n+m_1+m_2$.
	Thus, $(f_i + i) - (f_j + j)$ can take at most $2\cdot\frac{n + m_2 + m_1}{m_1} + 1 \in O(n^{\sfrac{2}{3}})$ different values.
	Finally, \cref{item:residue} implies that $((f_i - i) - (f_j - j) \bmod m_2) \in\{0, \ldots, m_1\}\cup\{m_2 - m_1, \ldots, m_2 - 1\}$ while also being of absolute value at most $n+m_1+m_2$.
	Thus, $(f_i - i) - (f_j - j)$ can take at most $2\cdot \frac{n + m_2 + m_1}{m_2} \cdot (2m_1 + 1) \in O(n^{\sfrac{2}{3}})$ different values.
	Plugging these findings into \eqref{eq:diffcount}, we conclude that $a^{1,\lambda}$ has at most $O(n^{\sfrac{2}{3}})$ different entries, as claimed.
\end{proof}

For completeness, let us also provide a simple existential proof of \cref{lem:shift-properties}.

\begin{proof}[Proof of \cref{lem:shift-properties}]
	To determine $f_j$, consider the elements of the set
	\[
		R \coloneqq\{ (j \bmod m_2), (j \bmod m_2) + 1, \dots, (j \bmod m_2) + m_1 - 1 \} \enspace.
	\]
	All of these elements satisfy \cref{item:absvalue,item:residue}.
	Furthermore, $R$ contains $m_1$ many consecutive integers, hence for precisely one $r\in R$, $j+r$ is divisible by $m_1$.
	We set $f_j$ equal to this $r$, so that \cref{item:divisible} is satisfied as well.
\end{proof}

\begin{remark}
	While it is not relevant for our application, note that the values $f_j$ guaranteed by \cref{lem:shift-properties} can be computed explicitly by the formula
	\[
		f_j = (j \bmod m_2) + m_1 - (((j \bmod m_2) + j) \bmod m_1) \enspace.
	\]
\end{remark}

\section{Hardness of separation over label constraints}\label{sec:separation}

In this section, we show that it is \NP-complete to decide whether a given point $x$ satisfies $x\notin Q_{(G, R)}$, i.e., \cref{thm:separation}.
Note that this non-containment problem is in \NP{} because any violated constraint is a suitable certificate for non-membership.
\NP-completeness of the non-containment problem directly implies that the containment problem, i.e., the decision problem of whether $x\in Q_{(G, R)}$, is also \NP-hard.

We reduce from \cubicmaxcut, which is the unweighted version of the \maxcut problem where the input graph is cubic, i.e., every vertex has degree three.%
\footnote{%
Our reduction is quite general and can easily be adapted to work with other variants of \maxcut.
We decided to use \cubicmaxcut as the source problem for our reduction because it is a well-known variant of \maxcut and allows us to keep the presentation of our reduction simple.%
}
We formally define \cubicmaxcut as follows, where for a graph $G=(V,E)$ and a set $S\subseteq V$, we let $\delta_G(S)$ denote the set of edges with precisely one endpoint in $S$.

\begin{mdframed}[userdefinedwidth=0.95\linewidth]
	\linkdest{prb:cubic-max-cut}{\textbf{\cubicmaxcut:}}
	Given a cubic graph $G=(V,E)$ and $k\in \mathbb{Z}_{\geq 0}$, decide if $\exists S\subseteq V$ with $|\delta_G(S)|\geq k$.
\end{mdframed}
\cubicmaxcut is known to be \NP-complete \cite{bermanTighterInapproximabilityResults1999}.
As stated above, to show \NP-completeness of the non-containment problem, we will reduce from \cubicmaxcut.
In our reduction, cuts in an instance of the \cubicmaxcut problem will correspond to potentially violated constraints in the non-containment problem.

\begin{proof}[Proof of \cref{thm:separation}]
	Let $G=(V,E)$ be a cubic graph and let $k\in \mathbb{Z}_{\geq 0}$.
	Our goal is to construct an instance $(\widehat G, \widehat R, x)$ of the non-containment problem such that $x\notin Q_{(\widehat G, \widehat R)}$ if and only if $(G, k)$ is a yes-instance of \cubicmaxcut.
	To this end, we define an auxiliary graph $\widehat G=(\widehat V, \widehat E)$ by
	\[
		\widehat V \coloneqq \bigcup_{v\in V}\{v^+, v^-\}
		\quad\text{and}\quad
		\widehat E \coloneqq \big\{\{v^+, v^-\}\colon v\in
			V\big\} \cup
			\left(\bigcup_{\{v,w\}\in E} \{\{v^+,w^-\},\{v^-,w^+\}\} \right)
	\]
	and a set $\widehat R\subseteq \widehat E$ of red edges by
	\[
		\widehat R \coloneqq \big\{\{v^+, v^-\}\colon v\in V\big\}\enspace.
	\]
	In other words, we duplicate the vertex set $V$ and get two copies $V^+, V^-$.
	The red edges $\widehat R$ are those connecting the two copies corresponding to the same original vertex.
	The remaining edges are obtained by connecting, for every edge $e=\{v,w\}$ in $E$, the copy $v^+$ of one endpoint to the copy $w^-$ of the other endpoint and, symmetrically, the copy $w^+$ to the copy $v^-$.
	Observe that all edges have one endpoint in $V^+$ and the other endpoint in
	$V^-$, hence $\widehat{G}$ is bipartite.
	An example of this construction can be found in \cref{fig:reduction-graph}.

	\begin{figure}[ht]
		\begin{center}
\begin{tikzpicture}[scale=0.4]
	\begin{scope}
		\begin{scope}[every node/.style={node_style}]
			\node (c) at (6.60,3.0) {};
			\node (d) at (3.60,6.0) {};
			\node (a) at (0.60,3.0) {};
			\node (b) at (3.60,0.0) {};
			\node (e) at (3.6, 2) {};
			\node (f) at (3.6, 4) {};
		\end{scope}

		\begin{scope}
			\node[right=2pt] at (c) {$c$};
			\node[above=2pt] at (d) {$d$};
			\node[left=2pt] at (a) {$a$};
			\node[below=2pt] at (b) {$b$};
			\node[right=2pt] at (e) {$e$};
			\node[left=2pt] at (f) {$f$};
		\end{scope}

		\begin{scope}[
				every path/.style={
					standard_line,
				},
			]
			\draw (c) -- (d);
			\draw (d) -- (a);
			\draw (d) -- (f) -- (e) -- (b);
			\draw (a) -- (b);
			\draw (b) -- (c);
			\draw (a) -- (e);
			\draw (c) -- (f);
		\end{scope}
	\end{scope}

	\begin{scope}[xshift=15cm, yshift=20mm]
		\begin{scope}[every node/.style={node_style}]
			\node (ap) at (0,6) {};
			\node (bp) at (0,4) {};
			\node (cp) at (0,2) {};
			\node (dp) at (0,0) {};
			\node (ep) at (0,-2) {};
			\node (fp) at (0,-4) {};

			\node (am) at (5,6) {};
			\node (bm) at (5,4) {};
			\node (cm) at (5,2) {};
			\node (dm) at (5,0) {};
			\node (em) at (5,-2) {};
			\node (fm) at (5,-4) {};
		\end{scope}

		\begin{scope}
			\node[left=1pt] at (ap) {$a^+$};
			\node[left=1pt] at (bp) {$b^+$};
			\node[left=1pt] at (cp) {$c^+$};
			\node[left=1pt] at (dp) {$d^+$};
			\node[left=1pt] at (ep) {$e^+$};
			\node[left=1pt] at (fp) {$f^+$};

			\node[right=2pt] at (am) {$a^-$};
			\node[right=2pt] at (bm) {$b^-$};
			\node[right=2pt] at (cm) {$c^-$};
			\node[right=2pt] at (dm) {$d^-$};
			\node[right=2pt] at (em) {$e^-$};
			\node[right=2pt] at (fm) {$f^-$};
		\end{scope}

		\begin{scope}[
				every path/.style={
					standard_line,
					black
				},
			]
			\draw (ap) -- (bm);
			\draw (am) -- (bp);
			\draw (ap) -- (dm);
			\draw (am) -- (dp);
			\draw (ap) -- (em);
			\draw (am) -- (ep);

			\draw (bp) -- (cm);
			\draw (bm) -- (cp);
			\draw (bp) -- (em);
			\draw (bm) -- (ep);

			\draw (cp) -- (dm);
			\draw (cm) -- (dp);
			\draw (cp) -- (fm);
			\draw (cm) -- (fp);

			\draw (dp) -- (fm);
			\draw (dm) -- (fp);

			\draw (ep) -- (fm);
			\draw (em) -- (fp);
		\end{scope}

		\begin{scope}[
				every path/.style={
					standard_line,
					red
				},
			]
			\draw (ap) -- (am);
			\draw (bp) -- (bm);
			\draw (cp) -- (cm);
			\draw (dp) -- (dm);
			\draw (ep) -- (em);
			\draw (fp) -- (fm);
		\end{scope}
	\end{scope}
\end{tikzpicture}
 		\end{center}
		\caption{Constructing an auxiliary graph $\widehat G=(\widehat V, \widehat E)$ with red edges $\widehat R$ (right) from an original graph $G$ (left).
		We define a vector $x$ on the edges of the auxiliary graph by assigning an $x$-value of $1-3\alpha$ to each red edge and an $x$-value of $\alpha$ to all other edges, where $\alpha\coloneqq\frac{1}{2(|E|-k+1)}$.}
		\label{fig:reduction-graph}
	\end{figure}
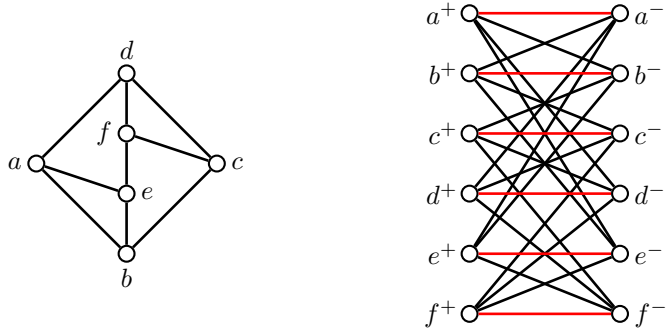

	Next, we define a point $x\in\mathbb{R}^{\widehat E}$ that we use in the non-containment problem.
  To this end, we let
  \begin{equation*}
		\alpha \coloneqq \frac{1}{2(|E| - k + 1)}\enspace,
	\end{equation*}
	and define $x$ by setting, for $e\in \widehat E$,
	\[
		x_e =
		\begin{cases}
			1-3\alpha & \text{if } e = \{v^+, v^-\} \in \widehat R \text{ for some } v\in V\\
			\alpha & \text{if } e \in \widehat E \setminus \widehat R
		\end{cases}\enspace.
	\]
	We can assume without loss of generality that $\alpha \leq \frac{1}{6}$, as otherwise we would have $|E|-k = O(1)$, in which case \cubicmaxcut can trivially be solved in polynomial time because there are only polynomially many edge sets with more than $k$ edges.
	Hence, every red edge has an $x$-value of at least $\sfrac{1}{2}$.
	Also, note that $x$ is in the perfect matching polytope of $\widehat G$ because there is precisely one unit of $x$-value on the edges incident to every vertex, i.e., degre constraint are satisfied.

	This finishes the construction of the instance $(\widehat G, \widehat R, x)$ of the non-containment problem.
	Recall that the corresponding label constraint relaxation is defined as
	\[
	Q_{(\widehat{G},\widehat R)}= \left\{
		x \in \mathbb{R}_{\geq 0}^{\widehat{E}} \colon
			\begin{array}{rl}
				x(\delta(u)) = 1 & \forall u\in \widehat V\\
				x(E_L)  \geq 1 & \forall L \in \Lall(\widehat{G})
			\end{array}
	\right\}\enspace.
	\]
	The containment problem at hand is to determine whether $x\notin Q_{(\widehat{G},\widehat R)}$.
	As degree constraints are satisfied by construction, $x\notin Q_{(\widehat{G},\widehat{R})}$ is equivalent to the existence of a labeling $L\in \Lall(\widehat G)$ with $x(E_L) < 1$.
	We first observe that $x$ satisfies the constraints defined by labelings $L$ with at least one red edge in $E_L$.

	\begin{claim}
		For all $L\in \Lall(\widehat{G})$ with $E_{L} \cap \widehat R \neq \emptyset$, we have $x(E_L) \geq 1$.
	\end{claim}

	\begin{proof}[Proof of claim]
		Assume that $L\in \Lall$ satisfies $E_{L} \cap \widehat R \neq \emptyset$, i.e., there are red edges in $E_L$.
		Recall that the red edges in $\widehat{G}$ are precisely the edges $\{v^+, v^-\}$, for $v\in V$.
		For them to be in $E_L$, we must have $L(v^+) \neq L(v^-)$.
		Note that $n=|V|$ is even because $G$ is a cubic graph, and every graph has an even number of vertices of odd degree.
		Hence, $|L^{-1}(1)|$ is even, too, which implies that there can only be an even number of red edges in $E_L$.
		As there are some red edges in $E_L$, it must thus contain at least two of them.
		Because the $x$-value on every red edge is at least $\sfrac{1}{2}$, we have $x(E_L) \geq \sum_{e\in E_L \cap \widehat R} x_e \geq 2\cdot \frac12 = 1$.
	\end{proof}

	Thus, only labelings $L$ for which no red edges appear in $E_L$, i.e., where $L(v^+) = L(v^-)$ for all $v\in V$, can lead to potentially violated constraints.
	Observe that there is a one-to-one correspondence between such labelings and cuts in $G$:
	from a labeling $L$ with $E_L\cap \widehat R=\emptyset$, we can construct a cut $S\subseteq V$ by setting $S \coloneqq \{v\in V\colon L(v^+) = L(v^-) = 1\}$; from a cut $S \subseteq V$ we can construct a labeling $L$ with $E_L\cap \widehat R=\emptyset$ by setting $L(v^+) = L(v^-) = 1$ for all $v\in S$ and $L(v^+) = L(v^-) = 0$ for all $v\in V\setminus S$.
	Beyond that, we can relate the $x$-value of the edges in $E_L$ to the number of edges in the corresponding cut $S$ as follows.
	
	Consider an edge $e=\{v,w\}\in E\setminus\delta_G(S)$, i.e., an edge of $G$ that is not in the cut defined by $S$.
	By construction, $e$ corresponds to two edges $\{v^+, w^-\}$ and $\{v^-, w^+\}$ of $\widehat G$ in $\widehat E \setminus \widehat R$.
	By definition of $S$, we have $L(v^+) = L(w^-) = L(v^-) = L(w^+)$.
	In particular, $\{v^+, w^-\}, \{v^-,w^+\}\in E_L$.
	Vice versa, for every edge $e = \{v^+, w^-\}\in E_L$, we have that $\{v,w\}\in E\setminus\delta(S)$.
	Consequently,
	\[
		x(E_L) = \alpha \cdot |E_L| = \alpha \cdot  2|E \setminus \delta(S)| = 2\alpha(|E| - |\delta(S)|)\enspace.
	\]
	Thus, there is a labeling $L$ on $\widehat{G}$ with $x(E_L)<1$ if and only if there is an $S\subseteq V$ such that $2\alpha (|E| - |\delta(S)|) < 1$.
	The latter can be rearranged to
	\[
		|\delta(S)| > |E| - \frac{1}{2\alpha} = k - 1 \enspace.
	\]
	We thus obtain that there is a violated label constraint if and only if the original graph $G$ has a cut of size at least $k$.
	This finishes the reduction and thus the proof.
\end{proof}

\section{An explicit example where label constraints do not suffice}
\label{sec:explicit-example}

We give an explicit example of a bipartite graph $G=(V,E)$ and a set of red edges $R\subseteq E$ for which the label constraint relaxation $Q_{(G,R)}$ of \textcite{jia_2023_exactMatching} is not equal to the odd-red perfect bipartite matching polytope $P_{(G,R)}$.
Concretely, consider the bipartite graph $G$ together with the red edges $R$ presented in \cref{fig:example_graph}.

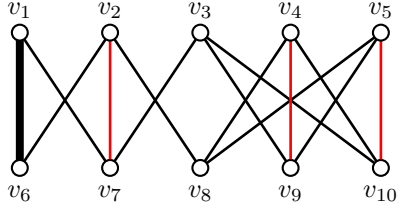
\begin{figure}[!ht]
	\begin{center}
\begin{tikzpicture}[scale=0.6]
	\newcommand{\xoffset}{2}
	\newcommand{\yoffset}{3}

	\foreach \i in {1,2,3,4,5} {
		\coordinate (\i) at (\i*\xoffset, \yoffset) {};
	}
	\foreach \i in {6,7,8,9,10} {
		\coordinate (\i) at (\i*\xoffset - 5*\xoffset, 0) {};
	}

	\begin{scope}[
		every path/.style={
			standard_line,
			black
		},
		]
		\draw[line width=3] (1) -- (6);
		\draw (1) -- (7);
		\draw (2) -- (8);
		\draw (2) -- (6);
		\draw (3) -- (7);
		\draw (3) -- (9);
		\draw (3) -- (10);
		\draw (4) -- (8);
		\draw (4) -- (10);
		\draw (5) -- (8);
		\draw (5) -- (9);

		\draw[red] (2) -- (7);
		\draw[red] (4) -- (9);
		\draw[red] (5) -- (10);
	\end{scope}

	\begin{scope}[every
		node/.style={node_style
		}]
		\foreach \i in {1,2,3,4,5} {
			\node (\i) at (\i) {};
		}
		\foreach \i in {6,7,8,9,10} {
			\node (\i) at (\i) {};
		}
	\end{scope}
		\foreach \i in {1,2,3,4,5} {
			\node[above=2pt] at (\i) {$v_{\i}$};
		}
		\foreach \i in {6,7,8,9,10} {
			\node[below=3pt] at (\i) {$v_{\i}$};
		}
\end{tikzpicture}
 	\end{center}
	\caption{Example of a graph $G=(V,E)$ with the indicated three red edges $R\subseteq E$ for which the polytope $Q_{(G,R)}$ is not integral.
		The point $y$ we consider to prove this has a weight of
		$\sfrac{2}{3}$ on the edge $(v_1,v_6)$ marked in bold, and
		$\sfrac{1}{3}$ on every other edge.}
	\label{fig:example_graph}
\end{figure}
\noindent
To be precise, we use $R\coloneqq \{\{v_2,v_7\}, \{v_4,v_9\}, \{v_5,v_{10}\}\}$.
We claim that $P_{(G,R)} \neq Q_{(G,R)}$ for this graph $G$.
To prove this, we define a point $y\in \mathbb{R}^E$ with value $\sfrac{1}{3}$ on every edge but the edge $\{v_1, v_6\}$, where it has value $\sfrac{2}{3}$ instead, i.e.,
\[
y(e) =
\begin{cases}
	\frac{1}{3} & \text{if } e \neq \{v_1,v_6\}\\
	\frac{2}{3} & \text{if } e = \{v_1,v_6\}
\end{cases}\quad\text{for all } e\in E\enspace.
\]
First of all, it is easy to see that $y$ does not lie in $P_{(G,R)}$.

\begin{lemma}\label{obs:y-not-in-P}
	The point $y$ cannot be expressed as a convex combination of
	incidence vectors of odd-red perfect matchings in $G$.
\end{lemma}

\begin{proof}
	Consider the edge $e = \{v_3, v_7\}$.
	Every perfect matching that contains this edge must also contain
	the edges $\{v_1, v_6\}$ and $\{v_2, v_8\}$.
	This implies that the perfect matching either uses both red edges
	$\{v_4, v_9\}$ and $\{v_5, v_{10}\}$, or none of them.
	Consequently, there is no odd-red matching containing $e$, finishing the
	proof, as $y(e)=\sfrac13$ is non-zero.
\end{proof}

However, we claim that $y$ is feasible for $Q_{(G,R)}$.

\begin{lemma}\label{claim:y-in-Q}
	The point $y$ is feasible for $Q_{(G,R)}$.
\end{lemma}

Together, \cref{obs:y-not-in-P,claim:y-in-Q} immediately imply that our construction is an example in which the label constraints do not suffice to describe the odd-red perfect matching polytope.
To show \cref{claim:y-in-Q}, we first repeat a lemma from \cite{jia_2023_exactMatching}, including a proof for completeness.
To this end, we recall from \cref{sec:candidate-descripton} the definitions
$$
\Lall(G) \coloneqq \{L \colon V \rightarrow \{0,1\} \text{
with } |L^{-1}(1)| \equiv n \pmod*{2}\}\enspace,
$$
and, for $L\in \Lall(G)$,
$$
  E_L = \{e\in E\setminus R \colon L(u) = L(v)\} \cup \{e \in
  R \colon L(u) \neq L(v)\}\enspace.
$$

\begin{lemma}[{\cite[Claim 5.2]{jia_2023_exactMatching}}]\label{obs:odd-cycles}
	Let $C$ be a cycle in $G$ with $|C\cap R|$ odd and let $L\in \Lall(G)$.
	Then $|E_L\cap C|$ is odd.
	In particular, the intersection is non-empty.
\end{lemma}

\begin{proof}
	Assume first that $L(v)=1$ for every vertex $v$ on $C$.
	Then we have $E_L\cap C = C\setminus R$.
	As $C$ is an even cycle, we have $|C\setminus R| \equiv|C\cap R| \pmod*{2}$.
	Together, this gives $|E_L\cap C| \equiv |C\cap R| \equiv 1 \pmod*{2}$, as desired.

	For the general case, note that the parity of $|E_L\cap C|$ does not change when flipping the label of a single vertex $v$ of $C$, as this flips the membership or non-membership in $E_L$ of both edges of $C$ incident to $v$.
	By applying this observation successively to all vertices labeled $0$, we get the statement for any labeling $L\in \Lall(G)$.
\end{proof}

We finish this section by proving \cref{claim:y-in-Q} to finalize the example.

\begin{proof}[Proof of \cref{claim:y-in-Q}]
	We first check that the degree constraints are fulfilled.
	To this end, note that all vertices but $v_1$ and $v_6$ have degree
	$3$, and $y$ has value $\sfrac{1}{3}$ on all of their incident edges.
	The vertices $v_1$ and $v_6$ only have two incident edges, but with
	$y$-value $\sfrac{1}{3}$ and $\sfrac{2}{3}$, respectively.
	Consequently, all degree constraints are satisfied.

	Proving that the partition constraints $y(E_L)\geq 1$ hold for all
	$L\in \Lall(G)$ is a bit more involved.
	Consider the cycle $C = (v_1, v_6, v_2, v_7)$ as well as the cycles
	$C_1 = (v_5, v_{10}, v_3, v_9)$, $C_2 = (v_4, v_{10}, v_5, v_8)$
	and $C_3 = (v_3, v_{10}, v_4, v_9)$.
	Each of these cycles contains exactly one red edge.
	By \cref{obs:odd-cycles}, regardless of how we choose $L\in
	\Lall(G)$, for each of the four cycles, one edge will be part of $E_L$.
	Note that $C$ is disjoint from any of the cycles $C_1, C_2, C_3$,
	and the latter three do not have a common edge, i.e., $C_1\cap C_2\cap C_3 = \emptyset$.
	Thus, we can conclude that for every $L\in \Lall(G)$, $E_L$
	contains at least $3$ edges.
	As $y$ has value at least $\sfrac{1}{3}$ on every edge, we must
	have $y(E_L)\geq 1$, as claimed.
\end{proof}

\end{document}